\PassOptionsToPackage{hyphens}{url}
\documentclass[12pt]{article}
\usepackage[letterpaper,margin=1in]{geometry}
\usepackage{amsmath,amssymb,amsthm,mathtools}
\usepackage{natbib}
\usepackage{setspace}
\usepackage{booktabs}
\usepackage{graphicx}
\usepackage{enumitem}
\usepackage{microtype}
\usepackage{array}
\usepackage{xurl}
\usepackage[colorlinks=true,linkcolor=blue,citecolor=blue,urlcolor=blue,
  pdftitle={Optimal Grading: A Unified Approach},
  pdfkeywords={grading, contests, rank information, performance information, fixed prizes, ironing, PAVA, mechanism design}]{hyperref}

\allowdisplaybreaks
\newtheorem{theorem}{Theorem}

\newtheorem{lemma}{Lemma}
\newtheorem{corollary}{Corollary}
\theoremstyle{definition}
\newtheorem{definition}{Definition}

\theoremstyle{remark}

\newcommand{\E}{\mathbb E}
\newcommand{\R}{\mathbb R}
\newcommand{\one}{\mathbf 1}
\newcommand{\co}{\operatorname{co}}
\newcommand{\dd}{\,\mathrm d}
\newcommand{\Wcal}{\mathcal W}
\newcommand{\Pbar}{\overline P}
\newcommand{\phibar}{\overline\varphi}
\newcommand{\Qtilde}{\widetilde Q}
\newcommand{\Vrank}{V^{\mathrm R}}
\newcommand{\Vperf}{V^{\mathrm P}}
\newcommand{\para}[1]{\medskip\noindent\textbf{#1}}

\title{Optimal Grading: A Unified Approach}
\author{Bin Liu\thanks{School of Management and Economics and Shenzhen Finance Institute,
		The Chinese University of Hong Kong, Shenzhen (CUHK-Shenzhen); Email: binliu@cuhk.edu.cn}
	\and Jingfeng Lu\thanks{Department of Economics, National University of Singapore; Email: ecsljf@nus.edu.sg}}

\date{August 24, 2026}

\begin{document}
\maketitle

\begin{abstract}
We develop a unified approach to optimal grading in an all-pay contest in
which a designer assigns a fixed vector of heterogeneous prizes to maximize
expected total effort. The approach covers two information regimes and
identifies a common principle: iron locally misordered incentive returns and
assign prizes assortatively across the resulting grades. Under rank-only
grading, assignments depend only on ordinal ranks. Ironing cumulative rank
coefficients---via the least concave majorant or the pool-adjacent-violators
algorithm---determines which adjacent ranks are pooled and which prizes are
randomized within each grade. Under performance-contingent grading,
assignments may depend on numerical effort. The optimum irons virtual ability,
forms endogenous type grades, and assigns prize blocks assortatively across
grades.
A failing grade below a minimum passing effort and a collection of effort
brackets implement the direct optimum while preserving full prize assignment.
\end{abstract}

\section{Introduction}\label{sec:introduction}

Grades, status categories, promotion tiers, titles, and badges are widely used to motivate effort and communicate relative standing. These devices do more than summarize performance. By determining which contestants are distinguished and which are treated alike, they change what contestants compete for. A fine grading system makes small differences in rank consequential; a coarse system places several ranks in a common category. The choice between these designs matters in schools, firms, professional contests, and online communities, particularly when the rewards attached to the available positions are fixed in advance.

A pioneering and influential analysis of this problem is \citet{MoldovanuSelaShi2007}. Building on the incomplete-information contest framework of \citet{MoldovanuSela2001}, they study an organization in which contestants care about relative status and the principal chooses the number and sizes of status categories to maximize output. Their analysis provides an important demonstration that the optimal degree of differentiation is fundamentally shaped by the distribution of ability: the top category is always a singleton, while the remainder of the organization may be finely or coarsely partitioned. More recently, \citet{Goel2025} develops a complementary theory of grading contests in which grades reveal information about ability and thereby determine market rewards. His analysis highlights the subtle incentive effects of information disclosure: more informative grading intensifies competition, but stronger competition may either encourage or discourage effort depending on the distribution of abilities and the cost technology.

Taken together, these papers establish grading as an important incentive-design problem rather than a purely informational choice and provide powerful characterizations of optimal grading in important environments. A natural next step is to seek a unified characterization that applies to an arbitrary continuous ability distribution and an arbitrary fixed prize vector, together with a constructive procedure for recovering the optimal grading from these primitives. It is also natural to ask how the solution changes when the designer can condition rewards on numerical performance rather than on ranks alone. To our knowledge, a common constructive characterization that addresses these questions has not yet been available.

This paper studies the problem in a unified fixed-prize framework. A designer has a weakly decreasing vector of prizes $\boldsymbol{v}=(v_1,\ldots,v_N)$ and allocates these prizes among contestants with privately known abilities and linear effort costs. Her goal is to maximize the expected total effort. The designer cannot change the prize values; she chooses how the fixed prizes are assigned. This formulation nests the pure-status problem of \citet{MoldovanuSelaShi2007}, after a harmless common shift of the singleton status values, and the linear-cost grading problem of \citet{Goel2025}, where the fixed prizes are the market values associated with exact ranks. It also applies directly to promotion slots, medals, tournament seats, and other settings in which rewards are committed in kind before the assignment rule is designed.

We distinguish two information regimes. Under \emph{rank-only grading}, the
assignment depends only on the ordinal ranking of efforts. The designer may
pool consecutive ranks into a common grade and randomize the corresponding
prizes within that grade, but she cannot use an absolute performance standard.
Under \emph{performance-contingent grading}, the numerical effort profile is
observable. The designer may then define a failing grade below a minimum
passing effort, treat an interval of effort levels as equivalent, or condition
the assignment more generally on absolute performance. Exclusion is not
allowed in either regime: every contestant receives one prize and all $N$
prizes are assigned. The second regime contains the first but supports
additional grading instruments that are unavailable when only ranks are
observed.

Our first set of results provides a complete characterization of the rank-only
optimum. Each rank has
a coefficient measuring the marginal contribution of prize value at that rank
to expected total effort. If these coefficients already decrease from the best
rank to the worst, the finest grading is optimal: rank $j$ receives the
$j$th-highest prize. When the coefficients are out of order, fine ranking
places too much incentive weight on some lower ranks relative to neighboring
higher ranks. The optimal response is to pool adjacent ranks until the average
coefficient of each grade decreases down the ranking. We formalize this idea
through a \emph{least-concave-majorant} construction and the
pool-adjacent-violators algorithm. The procedure is constructive, works for
every continuous ability distribution, and determines a canonical optimal
grading partition independently of the numerical levels of the fixed prizes.

Our second set of results solves the performance-contingent problem. Directly
optimizing over an equilibrium effort function is difficult because an optimal
rule may create flat regions, jumps, and unused effort intervals. We therefore
use a mechanism-design approach. The optimal direct mechanism irons virtual
ability: types for which finer distinctions do not generate sufficient
incentives are placed in a common performance grade, while types outside those
regions are separated. Better grades receive better blocks of the fixed
prizes, and the prizes within a grade are assigned uniformly. Because exclusion
is unavailable, a contestant can always exert zero effort and obtain at least
the lowest prize $v_N$. The optimal mechanism therefore leaves the lowest type
with utility $v_N$ and uses a failing grade below the minimum passing effort rather
than a reserve. We then implement the direct solution by a graded all-pay
contest with effort brackets.

Despite their different informational constraints, the two regimes share a
common analytical structure. Expected total effort is written as a linear
functional of the interim prize assignment; the fixed prize vector restricts
the feasible assignments; locally misordered incentive returns are ironed;
and prizes are assigned assortatively across the resulting blocks. What
differs is the object being ironed: rank coefficients under rank-only
information and virtual ability under performance information. Thus, the same
economic principle underlies both solutions. Grading should be fine where
additional distinctions create well-ordered incentive gains and coarse where
those gains are locally misordered. More information is therefore not
synonymous with finer grading. Numerical performance permits the designer to
define absolute performance categories, not to exclude contestants or
withhold prizes.

The results offer two practical lessons. First, a designer should not choose
the number of grades by convention alone: the appropriate boundaries depend
on the distribution of ability through the marginal effort returns of
neighboring ranks. Second, when absolute quality matters, a purely relative
system may be too restrictive because it rewards the same ranking whether
overall performance is high or low. A failing grade and effort brackets can
use absolute performance while preserving full prize assignment.

\paragraph{Related literature.}

Our paper belongs to the literature on prize design in all-pay contests. \citet{MoldovanuSela2001} study the optimal allocation of a prize budget across rewards, and \citet{MoldovanuSela2006} study contest architecture. With fixed heterogeneous prizes, \citet{LiuLu2014} show under regularity that assigning larger prizes to higher ranks maximizes effort. Our designer cannot change prize values; she chooses how a fixed vector is assigned. We remove regularity and separate the consequences of having only rank information from those of observing absolute performance. Related work studies entry fees, punishments, and negative prizes \citep{HammondEtAl2019,LiuLuWangZhang2018}, and optimal contests with nonlinear costs \citep{Zhang2024}.

Our paper also relates to grading and status.
\citet{MoldovanuSelaShi2007} study contests in which rewards are endogenous
status categories. \citet{DubeyGeanakoplos2010} show that coarse grades can
outperform fully revealing scores in an environment with status concerns.
\citet{Goel2025} studies optimal grading when a grade reveals information about
ability to a market. Our rank-only regime provides
a fixed-prize representation that nests the pure-status categories of
\citet{MoldovanuSelaShi2007} and the rank-value coarsenings of
\citet{Goel2025} with linear costs. Our performance regime instead permits grades to be defined
directly in absolute effort space while retaining full prize assignment.

Our techniques draw on the literature on ironing in mechanism design. \citet{Myerson1981} introduced ironing of virtual values, and \citet{Toikka2011} develops a general approach to ironing objectives without direct control of the underlying statistic. Reduced-form feasibility originates with \citet{Matthews1984} and \citet{Border1991}. The majorization and extreme-point methods of \citet{KleinerMoldovanuStrack2021} provide a broad geometric framework for such problems. Here majorization is generated by a fixed prize vector rather than a scalar budget, and the economic object being ironed is effort efficiency.

The rest of the paper proceeds as follows. Section~\ref{sec:model} presents the model and explains how the
status and grading models above are nested. Section~\ref{sec:rank} solves the
rank-only problem. Section~\ref{sec:performance} solves and implements
performance-contingent grading. Section~\ref{sec:conclusion} concludes. All
proofs are in the Appendix, except where noted.

\section{Model}\label{sec:model}

\subsection{Contestants, abilities, and fixed prizes}

There are $N\geq 2$ risk-neutral contestants. Contestant $i$ has ability
$t_i\in[a,b]$, where $a>0$. Abilities are independently and identically
distributed according to a continuously differentiable distribution function
$F$ with strictly positive density $f$ on $[a,b]$. A contestant of ability
$t_i$ who exerts effort $e_i\geq 0$ incurs effort cost $e_i/t_i$. Thus, if the
contestant receives an expected prize of value $w$, her expected payoff is
$w-e_i/t_i$.

The organizer has a fixed collection of $N$ nonnegative prizes
$\boldsymbol{v}=(v_1,\ldots,v_N)$, ordered so that
$v_1\geq v_2\geq\cdots\geq v_N\geq0$. We assume that at least one of the first
$N-1$ inequalities is strict; otherwise all prizes are identical and the
prize-design problem is trivial. The organizer chooses how these fixed prizes
are assigned among contestants so as to maximize expected total effort.

To describe rank-based prize assignments, let $\Pi(\boldsymbol{v})$ denote the set of all
permutations of $\boldsymbol{v}$. A \emph{deterministic rank-prize rule} is a vector
$\widetilde{\boldsymbol{w}}
=(\widetilde w_1,\ldots,\widetilde w_N)\in\Pi(\boldsymbol{v})$, where
$\widetilde w_j$ is the prize awarded to the contestant with the $j$th-highest
effort. Thus, a deterministic rank-prize rule simply specifies which of the fixed prizes
is attached to each rank. Ties in effort are broken uniformly at random.

The organizer may also randomize over deterministic rank-prize rules. A
\emph{stochastic rank-prize rule} is a probability distribution over
$\Pi(\boldsymbol{v})$ to which the organizer commits before contestants choose
their efforts. After efforts are exerted and ranks are determined, a
deterministic assignment is drawn according to this distribution and the
corresponding prizes are awarded. Since contestants are risk neutral, only the
expected prize associated with each rank matters for incentives. Hence a
stochastic rank-prize rule can be summarized by its expected rank-prize vector
$\boldsymbol{w}=(w_1,\ldots,w_N)$, where
$w_j=\E[\widetilde w_j]$ is the expected prize received by rank $j$.

Throughout, vectors are written as row vectors, and the superscript $T$ denotes
transpose. The set of expected rank-prize vectors that can be generated by
stochastic rank-prize rules is the \emph{prize permutahedron}
\[
\Wcal(\boldsymbol{v})
:=\co\Pi(\boldsymbol{v})
=
\left\{
\sum_{\boldsymbol{z}\in\Pi(\boldsymbol{v})}
\lambda_{\boldsymbol{z}}\boldsymbol{z}
:
\lambda_{\boldsymbol{z}}\geq0
\text{ for all }\boldsymbol{z}\in\Pi(\boldsymbol{v}),
\quad
\sum_{\boldsymbol{z}\in\Pi(\boldsymbol{v})}
\lambda_{\boldsymbol{z}}=1
\right\}.
\]
Thus, every $\boldsymbol{w}\in\Wcal(\boldsymbol{v})$ is a convex combination
of deterministic assignments of the fixed prize vector
$\boldsymbol{v}$. Conversely, every such convex combination can be implemented
by an appropriate lottery over deterministic rank-prize rules.

\subsection{Regime R: rank-only information}

In the rank-only regime, the prize assignment rule can depend only on the
ordinal ranking of efforts. Let
$\boldsymbol{w}=(w_1,\ldots,w_N)$ denote the expected rank-prize vector, where
$w_j$ is the expected prize received by the contestant at rank $j$. We require
higher ranks to receive weakly higher expected prizes, so that
\begin{equation}\label{eq:def_w_arrow}
	\boldsymbol{w}\in\Wcal^{\downarrow}(\boldsymbol{v})
	:=
	\left\{
	\boldsymbol{w}\in\Wcal(\boldsymbol{v})
	:
	w_1\geq w_2\geq\cdots\geq w_N
	\right\}.
\end{equation}
Thus, a higher-ranked contestant receives a weakly higher expected prize.

The monotonicity requirement is imposed on the expected rank-prize vector
$\boldsymbol{w}$. An individual deterministic prize assignment in the support
of the underlying lottery need not itself be monotone in rank.

Throughout Regime R, we evaluate a nonconstant rank-prize rule at its symmetric
equilibrium in strictly increasing strategies and a constant rule at the
zero-effort equilibrium. Lemma~\ref{lem:rankequilibrium} establishes uniqueness
within the former class. All rank-only value and design statements below refer
to this equilibrium convention.

\para{Rank grading and adjacent coarsening.}
Such prize allocation rules are closely related to the status/grading designs
considered in the literature \citep{MoldovanuSelaShi2007,Goel2025}. A
\emph{fixed rank grading} chooses one partition of the ranks
$\{1,\ldots,N\}$ into consecutive blocks. Specifically, for some
$0=k_0<k_1<\cdots<k_M=N$, the blocks are
$\mathcal B_m=\{k_{m-1}+1,\ldots,k_m\}$, $m=1,\ldots,M$. Under the finest partition,
each rank forms a singleton block and rank $j$ receives prize $v_j$. If
$\mathcal B_m=\{\ell,\ldots,r\}$ is a non-singleton block, these ranks are treated as
equivalent and the prizes $v_\ell,\ldots,v_r$ are assigned uniformly at random
among the contestants occupying them. Hence every rank in the block receives
expected prize
$(r-\ell+1)^{-1}\sum_{j=\ell}^r v_j$. Thus, a fixed rank grading preserves the fixed prize vector $\boldsymbol{v}$ but coarsens the ranking by pooling consecutive ranks and randomizing the corresponding prizes within each block. The word ``fixed'' refers to the
partition; within-block randomization remains part of the grading rule. A
\emph{randomized rank grading} is a lottery over fixed rank-grading
partitions. Lemma~\ref{lem:orderedpolytope} below shows that randomized rank
gradings generate exactly $\Wcal^{\downarrow}(\boldsymbol{v})$.

This interpretation shows that the rank-only regime nests two prominent models of grading and status contests.

\para{Contests for status.}
Consider the pure-status component of \citet{MoldovanuSelaShi2007}. Index the singleton ranks from best to worst and let $\boldsymbol{s}=(s_1,\ldots,s_N)$, where $s_j=N+1-2j$. Thus, $s_j$ is the status payoff of rank $j$ when every rank forms a separate category. If ranks $\{\ell,\ldots,r\}$ are pooled into one category, its status payoff in \citet{MoldovanuSelaShi2007} is $(N-r)-(\ell-1)=N+1-\ell-r$. This is exactly the average of the corresponding singleton status payoffs, since $(r-\ell+1)^{-1}\sum_{j=\ell}^r s_j=N+1-\ell-r$. Hence a status partition in \citet{MoldovanuSelaShi2007} is precisely a fixed rank grading of the fixed singleton-status vector $\boldsymbol{s}$.

To satisfy our maintained nonnegativity restriction, define instead $v_j=s_j+N-1=2(N-j)\geq0$. Adding the same constant to every rank payoff does not affect effort incentives under rank-only information, since it leaves all prize differences unchanged. Therefore, the pure-status partition problem of \citet{MoldovanuSelaShi2007} is a special case of Regime R.

\para{Optimal grading contests.}
The model of \citet{Goel2025} is also nested by Regime R. In \citet{Goel2025}, a contestant with marginal effort cost $\theta$ incurs cost $\theta e$; setting $t=1/\theta$ gives the cost $e/t$ used in our model. Under full rank revelation, let $\boldsymbol{v}^*=(v_1^*,\ldots,v_N^*)$ denote the fixed vector of market values associated with exact ranks. If consecutive ranks $\{\ell,\ldots,r\}$ receive the same grade, Bayes plausibility implies that the grade value is $(r-\ell+1)^{-1}\sum_{j=\ell}^r v_j^*$. Thus, a deterministic grading scheme in \citet{Goel2025} is exactly a fixed rank grading of $\boldsymbol{v}^*$: consecutive ranks are partitioned into blocks and each block is assigned the average of its constituent full-rank values. Restricting Regime R to fixed rank gradings therefore recovers the linear-cost grading problem studied in \citet{Goel2025}.

\subsection{Regime P: performance information}

In practice, the designer may be able to condition the prize assignment rule
on contestants' absolute performance rather than only on their relative
ranks. Rank information discards cardinal information: two effort profiles can
generate exactly the same ranking even though the underlying performance
levels are very different. When numerical effort is observable, the designer
may define a failing grade below a minimum passing standard, pool effort levels
within prescribed intervals, and condition the assignment more generally on
the entire effort profile.

Formally, in the performance-contingent regime, the designer observes the
numerical effort profile $\boldsymbol{e}=(e_1,\ldots,e_N)$. For every profile,
the assignment rule selects a possibly randomized bijection between
contestants and the $N$ fixed prizes. Thus, each contestant receives exactly
one prize and every prize is used exactly once. Contestants below a specified
passing standard may be placed in a common bottom grade, but they remain
eligible for the corresponding bottom prize block. Regime P therefore permits
absolute performance categories and effort brackets, but not reserves,
exclusion, or prize withholding.

\subsection{Quantile notation and rank probabilities}

Let
$q(x)=F^{-1}(x)$ for $x\in[0,1]$,
and define virtual ability
\begin{equation*}
\varphi(x)
=q(x)-(1-x)q'(x)
=J(q(x)),
\qquad
J(t)=t-\frac{1-F(t)}{f(t)}.
\end{equation*}
For rank $j=1,\ldots,N$, define
\begin{equation*}
p_j(x)=\binom{N-1}{j-1}x^{N-j}(1-x)^{j-1}.
\end{equation*}
If all rivals use a strictly increasing effort strategy, a contestant at quantile $x$ occupies rank $j$ with probability $p_j(x)$. We have
$
\sum_{j=1}^{N}p_j(x)=1$ and $
\int_0^1p_j(x)\dd x=\frac1N.
$

For $\boldsymbol{w}\in\Wcal^{\downarrow}(\boldsymbol{v})$, the interim expected prize is
\begin{equation*}
P(x;\boldsymbol{w})=\sum_{j=1}^{N}w_jp_j(x).
\end{equation*}
The fully \emph{assortative} rule $\boldsymbol{w}=\boldsymbol{v}$, i.e., assigning prizes monotonically according to rank, induces
\begin{equation}\label{eq:Pbar}
\Pbar(x)=\sum_{j=1}^{N}v_jp_j(x).
\end{equation}
The derivative of $P(x;\boldsymbol{w})$ with respect to $x$ is
\begin{equation}\label{eq:Pprime}
P'(x;\boldsymbol{w})
=(N-1)\sum_{j=1}^{N-1}(w_j-w_{j+1})
\binom{N-2}{j-1}x^{N-j-1}(1-x)^{j-1}\ge0.
\end{equation}
It is strictly positive on $(0,1)$ whenever $\boldsymbol{w}$ is nonconstant (i.e., the $N$ coordinates in $\boldsymbol{w}$ are not all equal). The same formula applies to $\Pbar$ with $\boldsymbol{w}=\boldsymbol{v}$.

\section{Optimal grading with rank-only information}\label{sec:rank}

\subsection{Equilibrium effort and the rank coefficients}

\begin{lemma}[Rank-contest reduction]\label{lem:rankequilibrium}
Fix $\boldsymbol{w}=(w_1,\ldots,w_N)\in\Wcal^{\downarrow}(\boldsymbol{v})$. If $\boldsymbol{w}$ is nonconstant, the unique symmetric equilibrium within the class of strictly increasing strategies is
\begin{equation}\label{eq:rankeffort}
e_{\boldsymbol{w}}(t)
=tP(F(t);\boldsymbol{w})-\int_a^tP(F(s);\boldsymbol{w})\dd s-aw_N.
\end{equation}
The expected total effort is
\begin{equation}\label{eq:rankobjective}
\Vrank(\boldsymbol{w})
=N\int_0^1\varphi(x)P(x;\boldsymbol{w})\dd x-Na w_N
=\boldsymbol{\beta}\boldsymbol{w}^{T},
\end{equation}
where $\boldsymbol{\beta}=(\beta_1,\ldots,\beta_N)$ is defined componentwise by
\begin{equation}\label{eq:beta}
\beta_j
=N\int_0^1\varphi(x)p_j(x)\dd x
-Na\one\{j=N\}.
\end{equation}
Moreover,
\begin{equation}\label{eq:betasum}
	\sum_{j=1}^{N}\beta_j=0.
\end{equation}
If $\boldsymbol{w}$ is constant (i.e., $w_1=\cdots=w_N$), all types exert zero effort, and \eqref{eq:rankobjective} remains valid.
\end{lemma}

The rank-only design problem is therefore the finite linear program
\begin{equation}\label{eq:rankproblem}
	\Vrank
	=
	\max_{\boldsymbol{w}\in\Wcal^{\downarrow}(\boldsymbol{v})}
	\boldsymbol{\beta}\boldsymbol{w}^{T}
	=
	\max_{\boldsymbol{w}\in\Wcal^{\downarrow}(\boldsymbol{v})}
	\sum_{j=1}^{N}\beta_j w_j.
\end{equation}
The coefficient $\beta_j$ is thus the marginal contribution to expected total effort of an additional unit of expected prize value at rank $j$.

\subsection{The ordered prize permutahedron}

We begin by providing an equivalent characterization of the set of expected prize vectors $\Wcal^{\downarrow}(\boldsymbol{v})$, defined in \eqref{eq:def_w_arrow}, that the designer can choose. Clearly, $\Wcal^{\downarrow}(\boldsymbol{v})$ is a convex set. For $\boldsymbol{w}\in\R^N$, let
\[
S_k(\boldsymbol{w})=\sum_{j=1}^kw_j,
\qquad
V_k=\sum_{j=1}^kv_j
\]
denote the cumulative sums of the first $k$ prizes in $\boldsymbol{w}$ and
$\boldsymbol{v}$, respectively.

It is useful to define the expected prize vectors induced by fixed rank
gradings; these vectors will appear as the extreme points of
$\Wcal^{\downarrow}(\boldsymbol{v})$.

\begin{definition}[Adjacent coarsening of $\boldsymbol{v}$]\label{def:def_of_adjacent_coarsening}
	We say that $\boldsymbol{w}$ is an \emph{adjacent coarsening} of $\boldsymbol{v}$ if there exists a consecutive partition $0=k_0<k_1<\cdots<k_M=N$ such that, for each block $\mathcal B_m=\{k_{m-1}+1,\ldots,k_m\}$, all ranks in $\mathcal B_m$ receive the same expected prize, equal to the average of the original prizes in that block. That is, for every $j\in\mathcal B_m$,
	$
	w_j=\frac{1}{k_m-k_{m-1}}\sum_{r=k_{m-1}+1}^{k_m}v_r.
	$
\end{definition}
Equivalently, an adjacent coarsening is the expected prize vector induced by a
fixed rank-grading partition. The following result provides an equivalent
characterization of $\Wcal^{\downarrow}(\boldsymbol{v})$ and shows that its
extreme points are exactly the adjacent coarsenings of $\boldsymbol{v}$.

\begin{lemma}[Ordered prize permutahedron]\label{lem:orderedpolytope}
	A vector $\boldsymbol{w}=(w_1,\ldots,w_N)$ belongs to $\Wcal^{\downarrow}(\boldsymbol{v})$ if and only if
	\begin{equation}\label{eq:majorizationfinite}
		w_1\ge\cdots\ge w_N,
		\qquad
		S_k(\boldsymbol{w})\le V_k\quad(k=1,\ldots,N-1),
		\qquad
		S_N(\boldsymbol{w})=V_N.
	\end{equation}
	Moreover, the extreme points of $\Wcal^{\downarrow}(\boldsymbol{v})$ are exactly the adjacent coarsenings of $\boldsymbol{v}$. Consequently, every feasible expected prize vector can be generated by a lottery over fixed rank gradings.
\end{lemma}

Because the designer's objective function is linear and the extreme points of
$\Wcal^{\downarrow}(\boldsymbol{v})$ are fixed rank gradings (i.e., adjacent coarsenings of $\boldsymbol{v}$), there
always exists an optimal fixed rank grading.

\subsection{Rank ironing and the PAVA solution}

Now we are ready to characterize the optimum. Define
\begin{equation*}
B_0=0,
\qquad
B_k=\sum_{j=1}^k\beta_j,
\quad k=1,\ldots,N,
\end{equation*}
which are the cumulative sums of the first $k$ coordinates of $\boldsymbol{\beta}$.
By \eqref{eq:betasum}, $B_N=0$. Let $\widehat B$ be the \emph{least concave majorant} of the points
$(0,B_0),(1,B_1),\ldots,(N,B_N)$; that is, $\widehat B:[0,N]\to\mathbb{R}$
is the pointwise smallest concave function satisfying
$\widehat B(k)\ge B_k$ for every $k=0,\ldots,N$. The majorant touches the
original curve at both endpoints:
$\widehat B(0)=B_0$ and $\widehat B(N)=B_N$. Indeed, if either endpoint lay
strictly above the corresponding data value, lowering that endpoint to the
data value would preserve concavity and the majorant property, contradicting
minimality. Define its discrete slopes by
\[
\widehat\beta_j=\widehat B(j)-\widehat B(j-1),
\qquad j=1,\ldots,N,
\]
and write
$\widehat{\boldsymbol{\beta}}
=(\widehat\beta_1,\ldots,\widehat\beta_N)$.
Since $\widehat B$ is concave, its discrete slopes are weakly decreasing:
$\widehat\beta_1\ge\cdots\ge\widehat\beta_N$.

Let
\[
\widehat D
=
\{k\in\{0,\ldots,N\}:\widehat B(k)=B_k\}
\]
be the \emph{contact set}, namely the integer points at which the least concave
majorant coincides with the original cumulative-sum curve. Define the
\emph{kink set} by
\[
\widehat K
=
\left\{
k\in\widehat D\cap\{1,\ldots,N-1\}:
\widehat\beta_k>\widehat\beta_{k+1}
\right\}.
\]
Thus, a contact point is a kink only when the slope of $\widehat B$ strictly
falls there.

Suppose that $\widehat K$ contains $L-1$ points. Write
$\widehat K=\{k_1,\ldots,k_{L-1}\}$, where
$k_1<\cdots<k_{L-1}$, and set $k_0=0$ and $k_L=N$.\footnote{If
	$\widehat K=\varnothing$, then $L=1$ and there is a single block
	$I_1=\{1,\ldots,N\}$. If $\widehat K$ contains exactly one point $k_1$, then
	$L=2$ and the two blocks are $I_1=\{1,\ldots,k_1\}$ and
	$I_2=\{k_1+1,\ldots,N\}$.}
Each kink $k_\ell$ marks a boundary between ranks $k_\ell$ and
$k_\ell+1$. Thus, the $L-1$ kink locations partition the ranks into $L$
consecutive blocks
\[
I_\ell
=
\{k_{\ell-1}+1,\ldots,k_\ell\},
\qquad
\ell=1,\ldots,L.
\]

Because $\widehat B$ is affine on each interval
$[k_{\ell-1},k_\ell]$---equivalently, its graph on that interval is the
straight line segment joining
$(k_{\ell-1},\widehat B(k_{\ell-1}))$ and
$(k_\ell,\widehat B(k_\ell))$---its discrete slope
$\widehat\beta_j$ is constant for all $j\in I_\ell$. Moreover, since the endpoints of each block are
contact points, for every $j\in I_\ell$,
\[
\widehat\beta_j
=
\frac{\widehat B(k_\ell)-\widehat B(k_{\ell-1})}
{k_\ell-k_{\ell-1}}
=
\frac{B_{k_\ell}-B_{k_{\ell-1}}}
{k_\ell-k_{\ell-1}}
=
\frac{1}{|I_\ell|}
\sum_{r\in I_\ell}\beta_r.
\]

Hence each block collects adjacent ranks whose original coefficients are
replaced by their common block average. These block averages then strictly
decrease across blocks, because consecutive blocks are separated by kinks.
The same blocks can be computed by the pool-adjacent-violators algorithm (PAVA), described below in Section~\ref{subsec:pavaalgorithm}. Let $\boldsymbol{w}^{\mathrm{PAVA}}$ denote the prize vector associated with these blocks, with $j$th coordinate
\begin{equation}\label{eq:wPAVA}
	w_j^{\mathrm{PAVA}}
	=
	\frac{1}{|I_\ell|}
	\sum_{r\in I_\ell}v_r,
	\qquad
	j\in I_\ell.
\end{equation}

To understand the role of the majorant construction, begin with the original
coefficient vector $\boldsymbol{\beta}$. If the coefficients are already weakly
decreasing in rank, so that
$\beta_1\geq\beta_2\geq\cdots\geq\beta_N$, then the cumulative-sum sequence
$(B_k)_{k=0}^{N}$ is concave. Its least concave majorant therefore coincides
with the original cumulative-sum curve: $\widehat B(k)=B_k$ for every
$k=0,\ldots,N$. Equivalently,
$\widehat\beta_j=\beta_j$ for every $j=1,\ldots,N$, so no ironing is needed. In this case, since
both $\boldsymbol{\beta}$ and the prize vector $\boldsymbol{v}$ are weakly
decreasing, the rank-only problem in \eqref{eq:rankproblem} is solved by
$\boldsymbol{w}=\boldsymbol{v}$: the contestant at rank $j$ receives the
$j$th-highest prize $v_j$.

The problem becomes nontrivial when
$\boldsymbol{\beta}$ is not weakly decreasing. The least-concave-majorant construction is an \emph{ironing} procedure that restores
this monotonicity. It pools adjacent coefficients that violate the required
order and replaces them by their block averages, thereby producing the weakly
decreasing ironed vector
$\widehat{\boldsymbol{\beta}}$. Since both
$\widehat{\boldsymbol{\beta}}$ and $\boldsymbol{v}$ are weakly decreasing, the
largest ironed coefficient is paired with the largest prize, the
second-largest coefficient with the second-largest prize, and so on. This does
not mean that the designer directly replaces the original coefficients by the
ironed ones. Rather, the following theorem shows that
$\widehat{\boldsymbol{\beta}}\boldsymbol{v}^{T}$ equals the maximum of the
original problem and that this value is implemented by
$\boldsymbol{w}^{\mathrm{PAVA}}$, which averages the prizes within the blocks
generated by the ironing procedure.

The following theorem establishes the optimality of this construction.

\begin{theorem}[Rank-ironing solution]\label{thm:pava}
The vector $\boldsymbol{w}^{\mathrm{PAVA}}$, defined above in \eqref{eq:wPAVA}, solves the rank-only design problem \eqref{eq:rankproblem}. The optimal expected total effort is
\begin{equation*}
\Vrank
=\widehat{\boldsymbol{\beta}}\boldsymbol{v}^{T}
=\boldsymbol{\beta}\bigl(\boldsymbol{w}^{\mathrm{PAVA}}\bigr)^{T}
=\sum_{\ell=1}^{L}
\frac{\left(\sum_{j\in I_\ell}\beta_j\right)
      \left(\sum_{j\in I_\ell}v_j\right)}{|I_\ell|}.
\end{equation*}
Thus, an optimal fixed rank grading is obtained by averaging the prizes within each block generated by the rank-ironing procedure.
\end{theorem}

For an example based on an admissible ability distribution, let $N=5$
and let the quantile function be
$
q(x)=1+\frac{x}{10}+10x^7$, $x\in[0,1].
$
Since
$
q'(x)=\frac{1}{10}+70x^6>0,
$
this defines a continuously differentiable distribution with strictly positive
density on its support. Substituting into \eqref{eq:beta} gives
$
\boldsymbol{\beta}
=
\left(
\frac{142}{55},
\frac{47}{110},
\frac{13}{33},
\frac{707}{990},
-\frac{2038}{495}
\right).
$
The corresponding cumulative sums are
$
(B_0,B_1,\ldots,B_5)
=
\left(
0,
\frac{142}{55},
\frac{331}{110},
\frac{1123}{330},
\frac{2038}{495},
0
\right).
$
The least concave majorant consists of the straight line segments joining
$(0,0)$ to $\left(1,\frac{142}{55}\right)$,
$\left(1,\frac{142}{55}\right)$ to
$\left(4,\frac{2038}{495}\right)$, and
$\left(4,\frac{2038}{495}\right)$ to $(5,0)$. Its values at the integer
points are
\[
\bigl(\widehat B(0),\widehat B(1),\ldots,\widehat B(5)\bigr)
=
\left(
0,
\frac{142}{55},
\frac{4594}{1485},
\frac{5354}{1485},
\frac{2038}{495},
0
\right).
\]
Hence its discrete slopes are
$
\widehat{\boldsymbol{\beta}}
=
\left(
\frac{142}{55},
\frac{152}{297},
\frac{152}{297},
\frac{152}{297},
-\frac{2038}{495}
\right).
$
The contact and kink sets are therefore
$
\widehat D=\{0,1,4,5\}$ and
$\widehat K=\{1,4\}.
$
The corresponding blocks are
$
I_1=\{1\}$,
$I_2=\{2,3,4\}$, and
$I_3=\{5\}.
$

\begin{figure}[ht]
\centering
\includegraphics[width=0.72\textwidth]
{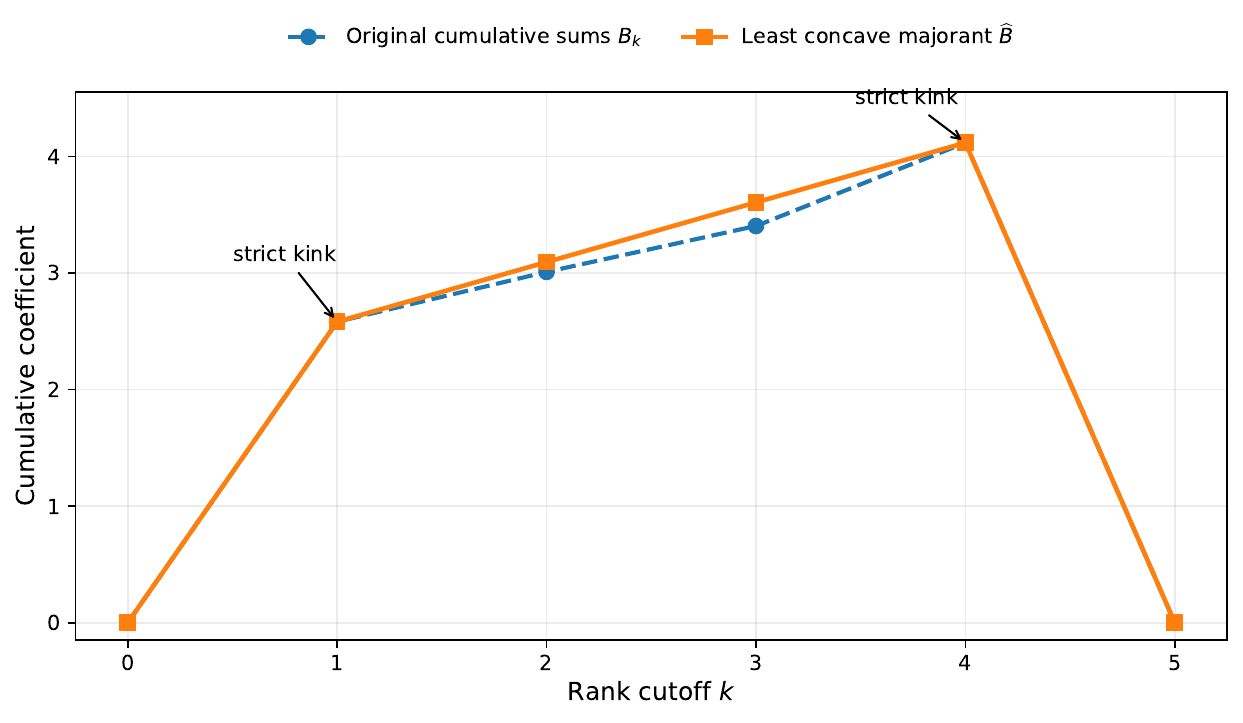}
\caption{The cumulative sums $B_k$ and their least concave majorant
$\widehat B$ in the five-rank example generated by
$q(x)=1+x/10+10x^7$. The strict kinks at $k=1$ and $k=4$
determine the blocks $\{1\}$, $\{2,3,4\}$, and $\{5\}$.}
\label{fig:pava-majorant-example}
\end{figure}

The associated expected prize vector is
\[
\boldsymbol{w}^{\mathrm{PAVA}}
=
\left(
v_1,
\frac{v_2+v_3+v_4}{3},
\frac{v_2+v_3+v_4}{3},
\frac{v_2+v_3+v_4}{3},
v_5
\right).
\]
Thus, rank $1$ remains a singleton, ranks $2$, $3$, and $4$ share the prizes
$v_2$, $v_3$, and $v_4$ uniformly, and rank $5$ receives $v_5$.

\subsection{The PAVA algorithm}\label{subsec:pavaalgorithm}

The least concave majorant can also be found without drawing the cumulative-sum curve. The \emph{pool-adjacent-violators algorithm} (PAVA) \citep[see][]{BarlowEtAl1972} processes the coefficients from the best rank to the worst rank and maintains provisional consecutive blocks. Each block is represented by the average of the coefficients it contains.

\begin{enumerate}[label=\textbf{Step \arabic*.},leftmargin=2.2cm]
\item Start with $\beta_1$ as the first provisional singleton block.

\item Move to the next coefficient and place it in a new singleton block. Compare the mean of the block on the left with the mean of the new block. If the left mean is weakly larger, the required decreasing order is satisfied, so keep the boundary provisionally and move to the next coefficient.

\item If the left mean is smaller than the right mean, the two blocks violate the required order. Pool them into a single block and replace their coefficients by their common average. For two blocks $I$ and $J$, the new mean is $\bigl(\sum_{j\in I\cup J}\beta_j\bigr)/|I\cup J|$.

\item After every pooling operation, move one block backward and compare the newly formed block with the block immediately to its left. If the left mean is again smaller, pool those two blocks as well. Continue moving backward until the block means are weakly decreasing. Then resume the forward pass with the next unprocessed coefficient.

\item Stop after $\beta_N$ has been processed and all adjacent block means are weakly decreasing. If two neighboring blocks have the same mean, combine them to obtain the maximal constant blocks of $\widehat{\boldsymbol{\beta}}$. Finally, average the corresponding prizes within each block according to \eqref{eq:wPAVA}.
\end{enumerate}

The backward check in Step 4 is essential. A new block average may be larger than the mean of an earlier block even when the original adjacent coefficients did not violate monotonicity. The algorithm therefore revisits earlier boundaries after every merge.

Return to the example above. Initially, each rank forms a singleton block,
with provisional means
$
\frac{142}{55}$,
$\frac{47}{110}$,
$\frac{13}{33}$,
$\frac{707}{990}$, and
$-\frac{2038}{495}.
$
The first three means are provisionally decreasing because
$
\frac{142}{55}
>
\frac{47}{110}
>
\frac{13}{33}.
$
The next comparison yields
$
\frac{13}{33}
<
\frac{707}{990},
$
so PAVA merges the singleton blocks $\{3\}$ and $\{4\}$. Their common mean is
$
\frac{1}{2}
\left(
\frac{13}{33}+\frac{707}{990}
\right)
=
\frac{1097}{1980}.
$
After this merger, the algorithm moves one block backward. Since
$
\frac{47}{110}
<
\frac{1097}{1980},
$
the block $\{2\}$ must be merged with $\{3,4\}$. The resulting block
$\{2,3,4\}$ has mean
$
\frac{1}{3}
\left(
\frac{47}{110}+\frac{13}{33}+\frac{707}{990}
\right)
=
\frac{152}{297}.
$
Finally,
$
\frac{142}{55}
>
\frac{152}{297}
>
-\frac{2038}{495},
$
so ranks $1$ and $5$ remain singleton blocks. The final blocks are therefore
$\{1\}$, $\{2,3,4\}$, and $\{5\}$, and
\[
\widehat{\boldsymbol{\beta}}
=
\left(
\frac{142}{55},
\frac{152}{297},
\frac{152}{297},
\frac{152}{297},
-\frac{2038}{495}
\right),
\]
which coincides with the discrete slope vector of the least concave majorant
shown in Figure~\ref{fig:pava-majorant-example}.

\para{Why the algorithm works.}
For any provisional block $I=\{p,\ldots,q\}$, its block mean
$\overline{\beta}(I):=|I|^{-1}\sum_{j\in I}\beta_j$ is the slope of the
straight line segment joining the cumulative-sum points
$(p-1,B_{p-1})$ and $(q,B_q)$. Since a concave function must have weakly
decreasing slopes, two adjacent blocks $I$ and $J$ with
$\overline{\beta}(I)<\overline{\beta}(J)$ cannot remain separate in the least
concave majorant. Pooling them replaces the two violating slopes by their
weighted average, which is exactly the slope of the straight line segment
joining the outer endpoints of $I\cup J$.

After a merger, the new block mean may exceed the mean of the block immediately
to its left, which explains why the algorithm must move backward and check
again. The procedure stops only when the means of all successive blocks are
weakly decreasing. At that point, the line obtained by connecting the
cumulative-sum points at the block boundaries is concave and lies weakly above
all the original points $(k,B_k)$. Moreover, any concave majorant must lie
weakly above each of these blockwise line segments. The resulting curve is
therefore the least concave majorant $\widehat B$.

Consequently, the final block means produced by PAVA are precisely the ironed
coefficients $\widehat\beta_j$. Adjacent final blocks with the same mean may be
combined to obtain the maximal blocks $I_1,\ldots,I_L$. Averaging the fixed
prizes over these same blocks gives
$\boldsymbol{w}^{\mathrm{PAVA}}$. Theorem~\ref{thm:pava} then formally shows that this
expected prize vector solves the rank-only design problem.

\section{Optimal grading with performance information}\label{sec:performance}

Once the designer observes numerical effort levels, she can endogenously
define performance categories based on effort profiles.
Directly analyzing the resulting equilibrium effort function can be
complicated because the optimal rule may create flat regions, jumps, and
unused effort intervals. We therefore adopt an indirect mechanism-design
approach. We first characterize the optimal direct mechanism and then show
that it can be implemented by a graded all-pay contest whose prize rule is
contingent on observed effort levels.

The direct-mechanism formulation yields an upper bound on the expected total
effort of every equilibrium outcome of an admissible performance-contingent
contest. Given any equilibrium,
possibly asymmetric or in mixed strategies, a direct mechanism asks contestants
to report their types and reproduces the possibly randomized effort and prize
outcome generated by the equilibrium strategies of the reported types. Truthful
reporting is Bayesian incentive compatible, and expected total effort is
unchanged. The resulting direct mechanism may initially be asymmetric.

Because types are i.i.d. and both feasibility and the objective are symmetric,
anonymity is without loss. Starting from any direct mechanism, draw a uniform
random permutation of contestant labels, apply the mechanism to the permuted
report profile, and then undo the permutation. This symmetrization preserves
feasibility, Bayesian incentive compatibility, the zero-effort deviation
constraint described below, and expected total effort. Consequently, every
equilibrium outcome has an anonymous direct-mechanism counterpart with the same
value. Conversely, Theorem~\ref{thm:implementation} implements the optimal
direct outcome by a graded all-pay contest. Thus, the direct problem gives an
upper bound that is attained within the contest class studied here.

\subsection{Direct mechanisms}

A direct mechanism asks each contestant to report a type and then specifies a
possibly randomized nonnegative effort requirement and assigns the $N$ fixed
prizes based on the report profile.
For every profile, each realized assignment uses every prize exactly once and
gives every contestant exactly one prize; the mechanism may randomize across
such assignments. Exclusion and prize withholding are not allowed.

For a symmetric mechanism, let $\Qtilde(t)$ denote the interim expected prize
of a truthful type $t$, let $E(t)$ denote her interim expected effort, and let
\[
U(t)=\Qtilde(t)-\frac{E(t)}{t}
\]
be her interim utility. In quantile space, write
$Q(x)=\Qtilde(q(x))$. The following result is standard in mechanism design
\citep[cf.][]{Myerson1981}, so we omit the proof.

\begin{lemma}\label{lem:IC}
A pair $(\Qtilde,E)$ with $E(t)\geq0$ for every $t$ is incentive compatible if
and only if $\Qtilde$ is nondecreasing and
\begin{equation}\label{eq:envelope}
E(t)
=t\Qtilde(t)-\int_a^t\Qtilde(s)\dd s-aU(a)
\end{equation}
for every $t\in[a,b]$, where the constant $U(a)=\Qtilde(a)-E(a)/a$ is the utility of the lowest type. Under incentive compatibility,
\[
U(t)
=\frac{aU(a)+\int_a^t\Qtilde(s)\dd s}{t},
\]
so ordinary individual rationality is equivalent to $U(a)\geq0$. Expected
total effort is
\begin{equation}\label{eq:directTE}
\mathrm{TE}
=N\int_0^1\varphi(x)Q(x)\dd x-NaU(a).
\end{equation}
\end{lemma}

The no-exclusion requirement is weakly stronger than ordinary individual
rationality. A contestant may always choose zero effort and must
still receive one of the fixed prizes, so this deviation yields at least
$v_N$. Hence every contest-implementable outcome must satisfy
$U(t)\geq v_N$ for every $t$. Under incentive compatibility,
$U'(t)\geq0$ almost everywhere, so it is enough to impose
$U(a)\geq v_N$. Since $U(a)$ enters \eqref{eq:directTE} with coefficient
$-Na$, the optimum sets
\begin{equation}\label{eq:bottomutility}
U(a)=v_N.
\end{equation}
The lowest type then exerts
\begin{equation}\label{eq:bottom-effort-general}
E(a)=a\bigl(\Qtilde(a)-v_N\bigr)\geq0.
\end{equation}
This is the lowest equilibrium effort. It is not an exclusion threshold:
efforts below it remain available and form a failing grade in the contest
implementation.

\subsection{Reduced-form feasibility}

Recall from \eqref{eq:Pbar} that
$\Pbar:[0,1]\to\mathbb{R}$ is defined by
$\Pbar(x)=\sum_{j=1}^{N}v_jp_j(x)$, where $p_j(x)$ is the probability that a
contestant at quantile $x$ attains rank $j$. Thus, $\Pbar(x)$ is the
contestant's interim expected prize under the fully assortative assignment, in
which rank $j$ receives prize $v_j$. The following notion will be helpful when
deriving a reduced-form feasibility bound.

\begin{definition}[Majorization]\label{def:majorization}
	Let $Q:[0,1]\to\mathbb{R}$ be a nondecreasing integrable function. We say that
	$Q$ is \emph{majorized} by $\Pbar$, denoted by $Q\preceq\Pbar$, if
	\begin{equation*}
		\int_x^1 Q(s)\dd s
		\leq
		\int_x^1 \Pbar(s)\dd s
		\quad\text{for every }x\in[0,1],
		\qquad
		\int_0^1 Q(s)\dd s
		=
		\int_0^1 \Pbar(s)\dd s
		=
		\frac{1}{N}\sum_{j=1}^{N}v_j.
	\end{equation*}
\end{definition}

The tail inequalities require that, for every quantile cutoff $x$, the
aggregate interim prize assigned to contestants with quantiles in $[x,1]$
does not exceed that under the fully assortative assignment. The equality of
the total integrals requires the entire fixed prize vector to be allocated. The following result provides the reduced-form feasibility bound used below.

\begin{lemma}[Feasibility bound]\label{lem:continuousfeas}
	The interim expected prize of every symmetric incentive-compatible mechanism
	is nondecreasing and satisfies $Q\preceq\Pbar$.
\end{lemma}

The tail inequalities have a direct interpretation. For any quantile cutoff
$x$, the contestants whose quantiles exceed $x$ cannot collectively receive
more than the same number of highest prizes. The fully assortative assignment,
whose interim prize function is $\Pbar$, attains this upper bound.

Combining Lemmas~\ref{lem:IC} and \ref{lem:continuousfeas}, the optimal value
in Regime P satisfies
\begin{equation}\label{eq:performanceprogram}
\Vperf
\leq
\overline V^{\mathrm P}
:=
\sup_Q
\left\{
N\int_0^1\varphi(x)Q(x)\dd x-Na v_N:
Q\text{ nondecreasing and }Q\preceq\Pbar
\right\}.
\end{equation}
We do not need to assume that every reduced form in this relaxation is
feasible. The optimizer below is generated by an explicit ex post assignment,
and Theorem~\ref{thm:implementation} implements it in a no-exclusion contest.

\subsection{Convex-envelope ironing in type space}

Define the integrated virtual ability
\begin{equation*}
\Phi(x)=\int_0^x\varphi(s)\dd s.
\end{equation*}
The following result provides some useful properties of $\Phi$.
\begin{lemma}\label{lem:Phi}
For every $x\in[0,1]$,
\begin{equation}\label{eq:Phiformula}
\Phi(x)=a-(1-x)q(x).
\end{equation}
In particular, $\Phi(0)=0$ and $\Phi(1)=a$.
\end{lemma}

Following the ironing approach of \citet{Myerson1981}, let $G$ be the
\emph{greatest convex minorant} of $\Phi$ on $[0,1]$. That is, $G$ is convex,
$G(x)\leq\Phi(x)$ for every $x\in[0,1]$, and every other convex function
$H$ satisfying $H(x)\leq\Phi(x)$ for all $x\in[0,1]$ also satisfies
$H(x)\leq G(x)$ for all $x\in[0,1]$. Equivalently, $G$ is the pointwise
largest convex function lying weakly below $\Phi$.

The minorant touches $\Phi$ at both endpoints:
$G(0)=\Phi(0)$ and $G(1)=\Phi(1)$. If either endpoint lay strictly below
$\Phi$, raising that endpoint to the corresponding value of $\Phi$ would
preserve convexity and the minorant property, contradicting maximality.

Because $G$ is convex, its right derivative is weakly increasing and
right-continuous. We therefore define
\begin{equation}\label{eq:phibar-definition}
	\phibar(x)
	=
	\begin{cases}
		G'_+(x),&x\in[0,1),\\
		G'_-(1),&x=1.
	\end{cases}
\end{equation}
We write $\dd\phibar$ for the associated Lebesgue--Stieltjes measure.

Define the \emph{type-ironing set} by
$
O
:=
\{x\in(0,1):G(x)<\Phi(x)\}.
$
Since both $G$ and $\Phi$ are continuous, $O$ is open relative to $(0,1)$.
It can therefore be written uniquely as an at most countable union of
pairwise disjoint open intervals:
\[
O=\bigcup_{m\in\mathcal M}(\ell_m,r_m),
\qquad
0\leq\ell_m<r_m\leq1.
\]
Each interval $(\ell_m,r_m)$ is a maximal connected component of $O$, meaning
that it is not properly contained in any larger interval lying entirely
inside $O$.\footnote{If $O=\varnothing$, then the index set $\mathcal M$ is
	empty, $G=\Phi$ on $[0,1]$, and $\phibar=\varphi$. In this
	case, no type ironing occurs.} On each component, $G$ is affine and
\begin{equation*}
\phibar(x)
=\frac1{r_m-\ell_m}
\int_{\ell_m}^{r_m}\varphi(s)\dd s.
\end{equation*}
Because $\varphi$ is continuous, $O$ is empty if and only if $\varphi$, equivalently $J$, is nondecreasing.

For implementation, associate with each component the operational grade cell
\begin{equation}\label{eq:type-ironing-cells}
I_m
=
\begin{cases}
[\ell_m,r_m),&r_m<1,\\
[\ell_m,1],&r_m=1.
\end{cases}
\end{equation}
This convention assigns every boundary quantile to exactly one grade cell. In
particular, if ironing begins at zero, the lowest type belongs to the bottom
passing grade.

Using these cells, define the following expected prize function:
\begin{equation}\label{eq:Qstar}
Q^*(x)=
\begin{cases}
\dfrac1{r_m-\ell_m}
\displaystyle\int_{\ell_m}^{r_m}\Pbar(s)\dd s,
&x\in I_m,\\[2ex]
\Pbar(x),&x\notin\displaystyle\bigcup_m I_m.
\end{cases}
\end{equation}
Thus, the fully assortative interim prize is averaged over every type-ironing
grade.

\para{Graded assortative rule.}
Given a reported ability profile
$\widehat{\boldsymbol{t}}=(\widehat t_1,\ldots,\widehat t_N)$, let
$\widehat x_i=F(\widehat t_i)$ denote contestant $i$'s reported quantile.
Each operational cell $I_m$ forms one non-singleton grade. Every reported
quantile $\widehat x_i$ outside $\bigcup_m I_m$ forms the point grade $\{\widehat x_i\}$; contestants
making the same point report are tied. The mechanism orders grades by their
locations in quantile space, ranks every contestant in a higher grade above
every contestant in a lower grade, and breaks ties within each grade uniformly.
It then assigns all fixed prizes $\boldsymbol{v}$ assortatively according to
this randomized graded ranking.

The following result shows that $Q^*$ is optimal.

\begin{theorem}[Performance-contingent optimum]\label{thm:performance}
The optimal value under performance-contingent grading with no exclusion is
\begin{equation}\label{eq:Vperf}
\Vperf
=\overline V^{\mathrm P}
=N\int_0^1\phibar(x)\Pbar(x)\dd x-Na v_N.
\end{equation}
This value is attained by $Q^*$.

Set $U(a)=v_N$ and, for ability $t$, require effort
\begin{equation}\label{eq:Estar}
E^*(t)
=t\Qtilde^*(t)-\int_a^t\Qtilde^*(s)\dd s-a v_N,
\qquad
\Qtilde^*(t)=Q^*(F(t)).
\end{equation}
The interim prize rule $Q^*$ is implemented by the graded assortative rule
defined above. Together with the effort schedule $E^*$, this rule is incentive
compatible, satisfies the no-exclusion zero-effort constraint, and maximizes
expected total effort. The lowest type receives utility $v_N$ and exerts
\begin{equation}\label{eq:emin}
E^*(a)=a\bigl(Q^*(0)-v_N\bigr).
\end{equation}
The effort schedule is constant on every nondegenerate grade and strictly
increasing between grades.
\end{theorem}

To interpret the assignment rule, suppose that a realized grade contains
$k$ contestants and that $h$ contestants belong to strictly higher grades.
The members of this grade then occupy ranks $h+1,\ldots,h+k$ and are ordered
uniformly at random within those ranks. Consequently, each receives a uniform
lottery over the prizes $v_{h+1},\ldots,v_{h+k}$. Thus, the mechanism ignores
reported performance differences within an ironing interval while preserving
the ordering of contestants across different grades.

The proof of Theorem~\ref{thm:performance} has two economically distinct
inequalities. For any feasible $Q$, the virtual-surplus term satisfies
\begin{equation}\label{eq:twostepbound}
\begin{aligned}
\int_0^1\varphi(x)Q(x)\dd x
&=
\int_0^1\phibar(x)Q(x)\dd x
-
\int_{[0,1]}\bigl(\Phi(x)-G(x)\bigr)\dd Q(x)\\
&\leq
\int_0^1\phibar(x)Q(x)\dd x
\leq
\int_0^1\phibar(x)\Pbar(x)\dd x.
\end{aligned}
\end{equation}
The first inequality says that a monotone assignment should be flat where raw
virtual ability must be ironed. The second says that once virtual ability is
monotone, the best feasible assignment is assortative in the majorization
order. The no-exclusion rent term $Na v_N$ is then subtracted from every
feasible design. The rule $Q^*$ makes both inequalities bind.

The construction of the graded assortative rule is equally important. In every
realization, contestants belonging to a grade occupy the same consecutive rank
positions that they occupy under full sorting; only their internal ordering is
randomized. Uniform tie-breaking therefore gives each grade member the
conditional average of the corresponding fixed prize block.

\begin{corollary}[Regular distributions]\label{cor:regular}
If $J$ is nondecreasing, then $O=\varnothing$, $Q^*=\Pbar$, and the optimal
performance design is the ordinary assortative rank rule. The lowest type
exerts
\[
E^*(a)=a\bigl(\Pbar(0)-v_N\bigr)=0,
\]
so no nontrivial failing region is needed.
\end{corollary}

\begin{corollary}[Performance-ironed rank coefficients]
\label{cor:performancecoeff}
Define
\begin{equation*}
\overline\beta_j
=N\int_0^1\phibar(x)p_j(x)\dd x
-Na\one\{j=N\}.
\end{equation*}
Then
\[
\Vperf=\sum_{j=1}^{N}\overline\beta_jv_j,
\qquad
\overline\beta_1\geq\cdots\geq\overline\beta_N,
\qquad
\sum_{j=1}^{N}\overline\beta_j=0.
\]
\end{corollary}

The raw rank coefficients $\boldsymbol{\beta}$ average the possibly
nonmonotone $\varphi$, whereas the performance-ironed coefficients average the
monotone $\phibar$. Both coefficient vectors subtract the unavoidable
bottom-prize rent through their final coordinate. Better performance grades
therefore receive better prize blocks even when rank-only PAVA pools several
ordinal ranks.

\subsection{Implementation by effort brackets}\label{subsec:implementation}

A \emph{graded all-pay contest} specifies a minimum passing effort
$e_{\min}\geq0$ and a finite or countable collection of disjoint passing-effort
brackets
$
[\underline e_m,\overline e_m)\subseteq[e_{\min},\infty).
$
Every contestant chooses an effort $e\geq0$. All efforts in
$[0,e_{\min})$ form the \emph{failing grade}; this grade is empty when
$e_{\min}=0$. A contestant in the failing grade is not excluded. If $k$
contestants fail, they occupy the bottom $k$ graded ranks and split the
associated bottom prize block uniformly. Efforts in the same passing bracket
are declared equivalent. Efforts in different brackets, and efforts outside
all brackets, are ordered by their numerical levels. Ties are broken uniformly,
and all $N$ prizes are assigned according to graded rank.

For each operational type grade $I_m$ associated with an ironing component
$(\ell_m,r_m)$, the optimal effort schedule is constant; denote this level by
$\widehat e_m$. Set
$\underline e_m=\widehat e_m$. If $r_m<1$, define
$
\overline e_m
=
\lim_{t\downarrow q(r_m)}E^*(t),
$
where the limit is taken from types immediately above the grade. If $r_m=1$,
set $\overline e_m=\infty$. The associated effort-grade cell is
$
\mathcal E_m=[\underline e_m,\overline e_m).
$
The right endpoint is excluded when $r_m<1$ because it is assigned to the
interval grade or singleton point immediately above $\mathcal E_m$.

Set the minimum passing effort equal to the lowest equilibrium effort,
\begin{equation*}
e_{\min}
=E^*(a)
=a\bigl(Q^*(0)-v_N\bigr).
\end{equation*}
The passing effort space $[e_{\min},\infty)$ is partitioned into the interval
grades $\mathcal E_m$ and singleton point grades. Specifically, every passing
effort level outside the interval grades forms its own singleton grade:
\[
\bigl\{\mathcal E_m:m\in\mathcal M\bigr\}
\;\cup\;
\left\{
\{e\}:
e\in[e_{\min},\infty)
\setminus\bigcup_{m\in\mathcal M}\mathcal E_m
\right\}.
\]
Thus, every passing effort not contained in an interval bracket, including
unused off-equilibrium effort levels, remains fully distinguished. The failing
grade $[0,e_{\min})$ lies below all passing grades. Hence, no effort is
prohibited and no contestant is excluded.

\begin{theorem}[Contest implementation]\label{thm:implementation}
In the graded all-pay contest constructed above, the strategy
$e(t)=E^*(t)$ is a symmetric Bayesian equilibrium. It induces the interim
prize $Q^*$, gives the lowest type utility $v_N$, assigns every prize, and
attains \eqref{eq:Vperf}.
\end{theorem}

The bracket commitment resolves the difficulty created by flat effort
schedules. A contestant cannot move from the bottom to the top of a tie by
adding an arbitrarily small amount of effort: every effort inside the same
bracket receives the same grade. Choosing an attained separating effort is
equivalent to reporting the corresponding type in the direct mechanism, while
choosing a bracket is equivalent to reporting any type in the associated
grade. Efforts strictly inside a bracket are dominated by its lower endpoint.
An effort in an unused gap below a passing bracket yields the limiting prize
from separating reports that approach the grade from below, but at a strictly
higher cost; the corresponding limiting incentive constraint rules it out. An
effort
below $e_{\min}$ receives the failing grade rather than exclusion, and zero
effort is the cheapest way to obtain that grade.

\subsection{The unified ironing principle}

The two regimes differ in what the designer observes and, consequently, in
the domain over which grading can operate. Nevertheless, their solutions
follow the same four steps: express expected total effort as a linear
functional of an interim prize assignment, characterize feasibility from the
fixed prize vector, iron locally misordered incentive returns, and match
better prizes assortatively to the resulting grades. Table~\ref{tab:unified}
summarizes the parallel.

\begin{table}[htbp]
\centering
\small
\caption{The unified structure of the two grading regimes}
\label{tab:unified}
\begin{tabular}{@{}>{\raggedright\arraybackslash}p{0.19\textwidth}
                    >{\raggedright\arraybackslash}p{0.35\textwidth}
                    >{\raggedright\arraybackslash}p{0.35\textwidth}@{}}
\toprule
& Rank-only grading & Performance-contingent grading \\
\midrule
Information
& Ordinal ranking of efforts
& Numerical effort profile \\
Reduced-form assignment
& Expected rank-prize vector $\boldsymbol{w}$
& Interim expected prize $Q$ \\
Feasibility
& Ordered permutahedron $\Wcal^{\downarrow}(\boldsymbol{v})$
& Monotonicity and $Q\preceq\Pbar$ \\
Raw incentive returns
& Rank coefficients $\beta_j$
& Virtual ability $\varphi(x)$ \\
Ironing
& Least concave majorant of cumulative rank coefficients (PAVA)
& Greatest convex minorant of integrated virtual ability \\
Grades
& Blocks of adjacent ranks
& Intervals of adjacent ability quantiles \\
Implementation
& Uniform prize randomization within rank blocks
& Assortative prize blocks, effort brackets, and a failing grade \\
\bottomrule
\end{tabular}
\end{table}

The fixed prize vector therefore plays the same allocative role in both
regimes, whereas the ability distribution determines the marginal incentive
returns and hence the grade boundaries. The resulting partitions need not
coincide: performance information permits absolute standards and effort
brackets that rank-only grading cannot use.

\section{Discussion and conclusion}\label{sec:conclusion}

This paper develops a unified approach to optimal grading with fixed
heterogeneous prizes. In both information regimes, expected total effort is
reduced to a linear assignment objective, locally misordered incentive returns
are ironed, and better prizes are matched assortatively to the resulting
grades. What the designer observes determines the object and domain of
ironing: ordinal rank coefficients under rank-only information and virtual
ability across quantiles under performance information.

When only ranks are observable, an absolute performance standard is not an
admissible instrument. The designer instead irons a finite sequence of rank
coefficients. PAVA identifies an optimal partition of consecutive ranks:
whenever adjacent block means violate the required decreasing order, the
blocks are merged, and the fixed prizes associated with each resulting block
are randomized uniformly among its ranks.

When numerical effort is observable, the designer can choose which absolute
performance distinctions matter while continuing to assign every prize. The
optimal mechanism irons virtual ability in quantile space, creates endogenous
type grades, matches fixed prize blocks assortatively across grades, and
randomizes within a grade. A failing grade below the minimum passing effort
makes the lowest type's utility equal to the unavoidable bottom-prize value
$v_N$. Effort brackets implement the direct mechanism even though the
equilibrium effort schedule may contain flat portions and jumps.

Several extensions remain natural. Risk aversion would make the full
distribution of prizes within a grade relevant rather than only its mean.
Convex effort costs would interact with both the rank and participation
margins. Allowing prizes to be destroyed would replace majorization by
submajorization and make exclusion an additional design instrument. Asymmetric
abilities would require agent-specific rank densities and reduced-form
constraints. The distinction between ironing ordinal rank returns and ironing
local type returns should remain useful in each extension.

\appendix
\section{Proofs for rank-only grading}\label{app:rankproofs}

\subsection{Proof of Lemma~\ref{lem:rankequilibrium}}

\begin{proof}
Fix $\boldsymbol{w}\in\Wcal^{\downarrow}(\boldsymbol{v})$ and write $\Qtilde(t)=P(F(t);\boldsymbol{w})$. Equation \eqref{eq:Pprime} shows that $\Qtilde$ is nondecreasing and is strictly increasing on $(a,b)$ if $\boldsymbol{w}$ is nonconstant.

Suppose rivals use a strictly increasing effort schedule $e(\cdot)$. A type $t$ who mimics type $s$ obtains
$
\Qtilde(s)-\frac{e(s)}{t}.
$
The first-order envelope condition for truthful choice is
$
e'(t)=t\Qtilde'(t).
$
The lowest type obtains the worst-rank value $w_N=\Qtilde(a)$ and thus chooses zero effort. Hence $e(a)=0$. Integrating the differential equation by parts gives
\begin{align*}
e(t)
&=\int_a^t s\dd \Qtilde(s)
=t\Qtilde(t)-a\Qtilde(a)-\int_a^t\Qtilde(s)\dd s\\
&=tP(F(t);\boldsymbol{w})-\int_a^tP(F(s);\boldsymbol{w})\dd s-aw_N,
\end{align*}
which is \eqref{eq:rankeffort}. Since $\dd e=t\dd \Qtilde$, the schedule is strictly increasing when $\boldsymbol{w}$ is nonconstant.

To verify global optimality, let $U(t)=\Qtilde(t)-e(t)/t$ denote the truthful payoff of type $t$, and let type $t$ mimic $s$. Using the effort formula,
\begin{align*}
U(t)-\left(\Qtilde(s)-\frac{e(s)}t\right)
&=\frac1t\int_s^t\bigl(\Qtilde(z)-\Qtilde(s)\bigr)\dd z\ge0.
\end{align*}
Thus truthful effort is optimal.

Expected effort per contestant, where $T\sim F$ denotes the contestant's ability, is
\begin{align*}
\E[e(T)]
&=\int_a^b\left[t\Qtilde(t)-\int_a^t\Qtilde(s)\dd s-aw_N\right]f(t)\dd t\\
&=\int_a^b\bigl[tf(t)-(1-F(t))\bigr]\Qtilde(t)\dd t-aw_N\\
&=\int_a^bJ(t)\Qtilde(t)f(t)\dd t-aw_N\\
&=\int_0^1\varphi(x)P(x;\boldsymbol{w})\dd x-aw_N.
\end{align*}
Multiplying by $N$ and using the linearity of $P$ gives \eqref{eq:rankobjective}--\eqref{eq:beta}. Finally,
\[
\sum_{j=1}^{N}\beta_j
=N\int_0^1\varphi(x)\dd x-Na.
\]
Since \eqref{eq:Phiformula} implies $\int_0^1\varphi=a$, the sum is zero.

If $\boldsymbol{w}$ is constant, the expected prize does not depend on rank, so zero effort is an equilibrium and is optimal for every type. The objective formula gives zero because $\sum_j\beta_j=0$.
\end{proof}

\subsection{Proof of Lemma~\ref{lem:orderedpolytope}}

\begin{proof}
By Rado's theorem on majorization
\citep[Proposition C.1, p.~162]{MarshallOlkinArnold2011},
$\boldsymbol{w}\in\co\Pi(\boldsymbol{v})$ if and only if
$\boldsymbol{w}\prec\boldsymbol{v}$, where
$\boldsymbol{w}\prec\boldsymbol{v}$ means that
$\boldsymbol{w}$ is majorized by $\boldsymbol{v}$.
Since both $\boldsymbol{w}$ and $\boldsymbol{v}$ are weakly decreasing, this is
equivalent to $S_k(\boldsymbol{w})\leq V_k$ for
$k=1,\ldots,N-1$ and $S_N(\boldsymbol{w})=V_N$, which proves the first
statement.

We next characterize the extreme points. We first claim that every adjacent coarsening belongs to
$\Wcal^{\downarrow}(\boldsymbol{v})$. To see this, consider an adjacent coarsening associated
with a consecutive partition $I_1,\ldots,I_L$ of $\{1,\ldots,N\}$. For each
block $I_\ell$, let $m_\ell=|I_\ell|$ and define the $m_\ell\times m_\ell$
averaging matrix
$
A_\ell=\frac{1}{m_\ell}\boldsymbol{1}_{m_\ell}\boldsymbol{1}_{m_\ell}^{T}.
$
Every entry of $A_\ell$ is $1/m_\ell$, so each of its rows and columns sums to
one. Hence $A_\ell$ is doubly stochastic. Let
$A=\operatorname{diag}(A_1,\ldots,A_L)$ be the corresponding block-diagonal
matrix. Then $A$ is also doubly stochastic, and the prize vector $\boldsymbol{w}$ corresponding to the adjacent coarsening can be
written as
$
\boldsymbol{w}=\boldsymbol{v}A.
$
Indeed, multiplying by $A_\ell$ replaces every coordinate of
$\boldsymbol{v}$ in block $I_\ell$ by the average of the coordinates in that
block. Because $A$ is doubly stochastic, $\boldsymbol{v}A$ is majorized by
$\boldsymbol{v}$ and therefore belongs to $\co\Pi(\boldsymbol{v})$ by Rado's
theorem. Moreover, since $\boldsymbol{v}$ is weakly decreasing and the blocks
are consecutive, the average prize in an earlier block is weakly larger than
the average prize in any later block. Thus $\boldsymbol{w}$ is weakly
decreasing. Finally, averaging within blocks preserves the total sum of the
prizes. Hence every adjacent coarsening belongs to
$\Wcal^{\downarrow}(\boldsymbol{v})$.

We next show that the extreme points of
$\Wcal^{\downarrow}(\boldsymbol{v})$ are exactly the adjacent coarsenings of
$\boldsymbol{v}$.
\paragraph{Step 1: every adjacent coarsening is an extreme point.}
Let $\boldsymbol{w}$ be an adjacent coarsening associated with the
consecutive partition
\[
I_\ell=\{k_{\ell-1}+1,\ldots,k_\ell\},
\qquad
0=k_0<k_1<\cdots<k_L=N.
\]
We have already shown that
$\boldsymbol{w}\in\Wcal^{\downarrow}(\boldsymbol{v})$. To prove that it is an
extreme point, suppose that
\[
\boldsymbol{w}
=
\lambda_1\boldsymbol{w}^{+}
+
\lambda_2\boldsymbol{w}^{-},
\qquad
\boldsymbol{w}^{+},\boldsymbol{w}^{-}
\in\Wcal^{\downarrow}(\boldsymbol{v}),
\]
where $\lambda_1,\lambda_2>0$ and $\lambda_1+\lambda_2=1$.
At the endpoint $k_\ell$ of each block, the adjacent coarsening preserves the
sum of the corresponding original prizes. Hence
\[
S_{k_\ell}(\boldsymbol{w})
=
\sum_{h=1}^{\ell}\sum_{j\in I_h}w_j
=
\sum_{h=1}^{\ell}\sum_{j\in I_h}v_j
=
V_{k_\ell}.
\]
It follows that
$
V_{k_\ell}
=
\lambda_1S_{k_\ell}(\boldsymbol{w}^{+})
+
\lambda_2S_{k_\ell}(\boldsymbol{w}^{-}).
$
Feasibility gives
$S_{k_\ell}(\boldsymbol{w}^{+})\leq V_{k_\ell}$ and
$S_{k_\ell}(\boldsymbol{w}^{-})\leq V_{k_\ell}$. Since both weights are
strictly positive, both inequalities must bind:
$
S_{k_\ell}(\boldsymbol{w}^{+})
=
S_{k_\ell}(\boldsymbol{w}^{-})
=
S_{k_\ell}(\boldsymbol{w})
=
V_{k_\ell}.
$
Subtracting the equalities at two successive block endpoints
$k_{\ell-1}$ and $k_\ell$ yields
$
\sum_{j\in I_\ell}w_j^{+}
=
\sum_{j\in I_\ell}w_j^{-}
=
\sum_{j\in I_\ell}w_j.
$
Thus, the three vectors have the same total prize within every block.

Now fix a block $I_\ell$.\footnote{If $I_\ell=\{j\}$ is a singleton, equality
	of the block sums immediately gives
	$w_j^{+}=w_j^{-}=w_j$, or equivalently $d_j=0$. The argument below is therefore
	needed only when $I_\ell$ contains at least two ranks.} Since $\boldsymbol{w}$ is constant on this block,
define
$
d_j=w_j^{+}-w_j$, $j\in I_\ell.
$
The convex-combination identity implies
$
w_j^{-}
=
w_j-\frac{\lambda_1}{\lambda_2}d_j.
$
For two consecutive ranks $j,j+1\in I_\ell$, monotonicity of
$\boldsymbol{w}^{+}$ gives
$
w_j+d_j\geq w_{j+1}+d_{j+1}.
$
Because $w_j=w_{j+1}$ within the block, this implies
$d_j\geq d_{j+1}$. Similarly, monotonicity of
$\boldsymbol{w}^{-}$ implies
\[
w_j-\frac{\lambda_1}{\lambda_2}d_j
\geq
w_{j+1}-\frac{\lambda_1}{\lambda_2}d_{j+1},
\]
and hence $d_j\leq d_{j+1}$. Therefore, $d_j=d_{j+1}$, so the deviation
$d_j$ is constant throughout $I_\ell$.

Moreover, equality of the block sums gives
$
\sum_{j\in I_\ell}d_j
=
\sum_{j\in I_\ell}(w_j^{+}-w_j)
=
0.
$
Since the deviations are constant within the block, their common value must
be zero. Applying this argument to every block gives
$\boldsymbol{w}^{+}=\boldsymbol{w}^{-}=\boldsymbol{w}$. Thus, every adjacent
coarsening is an extreme point.

\paragraph{Step 2: every extreme point is an adjacent coarsening of $\boldsymbol{v}$.}
Conversely, let $\boldsymbol{w}$ be an extreme point of
$\Wcal^{\downarrow}(\boldsymbol{v})$, and define the set of binding cumulative
constraints by
$D=\{k\in\{0,\ldots,N\}:S_k(\boldsymbol{w})=V_k\}$, where
$S_0(\boldsymbol{w})=V_0=0$. Notice that $0,N\in D$. We claim that
$\boldsymbol{w}$ is constant between every two consecutive elements of $D$.

Let $r<s$ be consecutive elements of $D$. If $s=r+1$, the set $I=\{r+1,\ldots,s\}$ contains
only one coordinate (i.e., the $s$th coordinate) and $\boldsymbol{w}$ is therefore trivially constant on $I$. Suppose instead that
$s\geq r+2$ and, contrary to the claim, that $\boldsymbol{w}$ is not constant
on $I=\{r+1,\ldots,s\}$. Because $r$ and $s$ are consecutive binding
indices, every cumulative constraint strictly between them is slack:
$S_k(\boldsymbol{w})<V_k$ for $k=r+1,\ldots,s-1$.

Partition the sequence $(w_{r+1},\ldots,w_s)$ into its maximal constant runs. A constant run is a consecutive set of coordinates on which
$\boldsymbol{w}$ takes the same value, and it is maximal if it cannot be
extended to the left or right while preserving that value. Let
$A=\{r+1,\ldots,p\}$ and $C=\{p',\ldots,s\}$ be the first and last such runs,
respectively. A run may consist of a single coordinate; the argument below
applies without change because then its cardinality is one. Since
$\boldsymbol{w}$ is not constant on $I$, there are at least two runs, so
$A$ and $C$ are disjoint.

For ease of exposition, we refer to the boundary between two consecutive
constant runs as a \emph{junction}. If
one run ends at rank $m$ and the next begins at rank $m+1$, then maximality
and the weak decrease of $\boldsymbol{w}$ imply $w_m>w_{m+1}$. Thus, every
junction between constant runs inside $I$ is a strict decrease.

The two outer boundaries of $I$ are also strict whenever they exist. Suppose
first that $r>0$. Since $\boldsymbol{w}$ is not constant on $I$, we have
$r+1<s$, so the cumulative constraint at $r+1$ is strictly slack; that is, $S_{r+1}(\boldsymbol{w})<V_{r+1}$. Moreover,
$S_r(\boldsymbol{w})=V_r$ and
$S_{r-1}(\boldsymbol{w})\leq V_{r-1}$. Therefore,
$w_r=S_r(\boldsymbol{w})-S_{r-1}(\boldsymbol{w})\geq v_r$, whereas
$w_{r+1}=S_{r+1}(\boldsymbol{w})-S_r(\boldsymbol{w})<v_{r+1}$. Since
$v_r\geq v_{r+1}$, it follows that $w_r>w_{r+1}$.

Similarly, suppose that $s<N$. Because $s-1>r$, the cumulative constraint at
$s-1$ is strictly slack; that is, $S_{s-1}(\boldsymbol{w})<V_{s-1}$. Hence
$w_s=S_s(\boldsymbol{w})-S_{s-1}(\boldsymbol{w})>v_s$, while feasibility at
$s+1$ gives
$w_{s+1}=S_{s+1}(\boldsymbol{w})-S_s(\boldsymbol{w})\leq v_{s+1}$. Since
$v_s\geq v_{s+1}$, we obtain $w_s>w_{s+1}$.

Define a perturbation vector $\boldsymbol{d}=(d_1,\ldots,d_N)$ by
\[
d_j=
\begin{cases}
	\frac{1}{|A|}, & j\in A,\\
	-\frac{1}{|C|}, & j\in C,\\
	0, & \text{otherwise}.
\end{cases}
\]
Thus, the perturbation raises every coordinate in the first run by the same
amount $1/|A|$ and lowers every coordinate in the last run by the same amount $1/|C|$. The
normalization ensures that the total increase on $A$ equals the total decrease
on $C$, so $\sum_{j=1}^{N}d_j=0$.

Let $\Delta_k=\sum_{j=1}^{k}d_j$. By construction,
$\Delta_k=0$ for $k\leq r$ and for $k\geq s$, while
$\Delta_k>0$ for every $r<k<s$. In particular, the perturbation changes only
the cumulative sums associated with the strictly slack constraints between
$r$ and $s$.

Fix $\lambda_1,\lambda_2>0$ such that $\lambda_1+\lambda_2=1$. For
$\varepsilon>0$, define
\[
\boldsymbol{w}^{+}
=
\boldsymbol{w}+\lambda_2\varepsilon\boldsymbol{d},
\qquad
\boldsymbol{w}^{-}
=
\boldsymbol{w}-\lambda_1\varepsilon\boldsymbol{d}.
\]
Then
$\lambda_1\boldsymbol{w}^{+}+\lambda_2\boldsymbol{w}^{-}
=\boldsymbol{w}$.

For sufficiently small $\varepsilon$, both perturbed vectors remain weakly
decreasing. Indeed, all coordinates within $A$ receive the same perturbation,
as do all coordinates within $C$, so equality within each constant run is
preserved. The only ordering inequalities that could become tighter occur at
the boundaries of $A$ and $C$. The interior boundaries are strict because
they are junctions between distinct maximal constant runs, and the outer
boundaries are strict by the preceding argument. Hence none of these
inequalities is reversed when $\varepsilon$ is sufficiently small.

The cumulative sums of the perturbed vectors satisfy
$S_k(\boldsymbol{w}^{+})
=S_k(\boldsymbol{w})+\lambda_2\varepsilon\Delta_k$ and
$S_k(\boldsymbol{w}^{-})
=S_k(\boldsymbol{w})-\lambda_1\varepsilon\Delta_k$. For every
$r<k<s$, define the slack in the $k$th cumulative constraint by
$\sigma_k:=V_k-S_k(\boldsymbol{w})$. Because $r$ and $s$ are consecutive
elements of $D$, we have $\sigma_k>0$ for all such $k$. Moreover,
$\Delta_k>0$ on this interval. Hence, by choosing
$\varepsilon\leq\min_{r<k<s}\sigma_k/(\lambda_2\Delta_k)$, we ensure that
$S_k(\boldsymbol{w}^{+})\leq V_k$ for every $r<k<s$. For $k\leq r$ or
$k\geq s$, we have $\Delta_k=0$, so the cumulative sums of
$\boldsymbol{w}^{+}$ are unchanged and remain feasible. Since
$\Delta_k\geq0$ for every $k$, the perturbation defining
$\boldsymbol{w}^{-}$ only weakly decreases its cumulative sums, and therefore
$S_k(\boldsymbol{w}^{-})\leq V_k$ for every $k<N$. Finally,
$\Delta_N=0$, so both perturbed vectors preserve the total-sum condition
$S_N(\boldsymbol{w}^{-})=S_N(\boldsymbol{w}^{+})=V_N$.
Choosing $\varepsilon$ small enough also to preserve the monotonicity
conditions discussed above therefore gives
$\boldsymbol{w}^{+},\boldsymbol{w}^{-}
\in\Wcal^{\downarrow}(\boldsymbol{v})$.

Because $\boldsymbol{d}\neq\boldsymbol{0}$, the vectors
$\boldsymbol{w}^{+}$ and $\boldsymbol{w}^{-}$ are distinct. We have therefore
constructed a nontrivial convex decomposition
$\boldsymbol{w}
=\lambda_1\boldsymbol{w}^{+}+\lambda_2\boldsymbol{w}^{-}$, contradicting the
extremality of $\boldsymbol{w}$. It follows that $\boldsymbol{w}$ must be
constant between every two consecutive elements of $D$.

Now let $r<s$ be consecutive elements of $D$, and let $c$ denote the common
value of $w_{r+1},\ldots,w_s$. Since the cumulative constraints bind at both
endpoints,
$(s-r)c=S_s(\boldsymbol{w})-S_r(\boldsymbol{w})
=V_s-V_r=\sum_{j=r+1}^{s}v_j$. Hence
$c=(s-r)^{-1}\sum_{j=r+1}^{s}v_j$. If $s=r+1$, this simply says that the
singleton block receives $c=v_{r+1}$.

Thus, on every interval between consecutive binding indices in $D$,
$\boldsymbol{w}$ replaces the corresponding coordinates of
$\boldsymbol{v}$ by their block average. Therefore,
$\boldsymbol{w}$ is an adjacent coarsening of $\boldsymbol{v}$. We conclude that the extreme points of
$\Wcal^{\downarrow}(\boldsymbol{v})$ are exactly the adjacent coarsenings of
$\boldsymbol{v}$.

Finally, the characterization in \eqref{eq:majorizationfinite} shows that
$\Wcal^{\downarrow}(\boldsymbol{v})$ is a compact polytope, so it is the
convex hull of its extreme points. Consequently, every feasible expected prize
vector is a convex combination of adjacent coarsenings and can
therefore be implemented by a lottery over fixed rank gradings.
\end{proof}

\subsection{Proof of Theorem~\ref{thm:pava}}

\begin{proof}
Consider any fixed rank grading with $M$ consecutive blocks. Write
$\mathcal I=\{J_1,\ldots,J_M\}$ and choose block boundaries
$0=\tau_0<\tau_1<\cdots<\tau_M=N$ such that
$J_\ell=\{\tau_{\ell-1}+1,\ldots,\tau_\ell\}$ for
$\ell=1,\ldots,M$. Let $\boldsymbol{w}^{\mathcal I}$ be the adjacent
coarsening of $\boldsymbol{v}$ generated by this partition, and let
$\boldsymbol{\beta}^{\mathcal I}$ be obtained by averaging the original
coefficients over the same blocks. Thus, for every $j\in J_\ell$,
$w_j^{\mathcal I}=|J_\ell|^{-1}\sum_{r\in J_\ell}v_r$ and
$\beta_j^{\mathcal I}=|J_\ell|^{-1}\sum_{r\in J_\ell}\beta_r$.
Blockwise averaging then implies that the total effort induced by the prize vector $\boldsymbol{w}^{\mathcal I}$ is
\begin{equation*}
	\boldsymbol{\beta}
	\bigl(\boldsymbol{w}^{\mathcal I}\bigr)^{T}
	=
	\boldsymbol{\beta}^{\mathcal I}\boldsymbol{v}^{T}
	=
	\sum_{J\in\mathcal I}
	\frac{
		\left(\sum_{j\in J}\beta_j\right)
		\left(\sum_{j\in J}v_j\right)
	}{
		|J|
	}.
\end{equation*}

Define the cumulative sums of the block-averaged coefficients by
$B_h^{\mathcal I}=\sum_{j=1}^{h}\beta_j^{\mathcal I}$, with
$B_0^{\mathcal I}=0$. Recall that in the main text, $B_h$ was defined as $\sum_{j=1}^{h}\beta_j$, with $B_0=0$. At every block boundary $\tau_\ell$, averaging preserves
the total sum of the coefficients in all completed blocks. Hence
$B_{\tau_\ell}^{\mathcal I}=B_{\tau_\ell}$ for
$\ell=0,\ldots,M$.

Now fix a block $J_\ell$ and an integer
$h\in\{\tau_{\ell-1},\ldots,\tau_\ell\}$. Since
$\beta_j^{\mathcal I}$ is constant over $J_\ell$, the cumulative sum
$B_h^{\mathcal I}$ varies at a constant rate between the two block
boundaries. More precisely,
\[
B_h^{\mathcal I}
=
\frac{\tau_\ell-h}{\tau_\ell-\tau_{\ell-1}}
B_{\tau_{\ell-1}}
+
\frac{h-\tau_{\ell-1}}{\tau_\ell-\tau_{\ell-1}}
B_{\tau_\ell}.
\]
Thus, $(h,B_h^{\mathcal I})$ lies on the straight line segment joining the
original cumulative-sum points
$(\tau_{\ell-1},B_{\tau_{\ell-1}})$ and
$(\tau_\ell,B_{\tau_\ell})$.

Because $\widehat B$ lies weakly above the original cumulative sums,
$B_{\tau_{\ell-1}}\leq\widehat B(\tau_{\ell-1})$ and
$B_{\tau_\ell}\leq\widehat B(\tau_\ell)$.
Moreover, concavity of $\widehat B$ implies that its graph lies weakly above
the straight line segment joining
$(\tau_{\ell-1},\widehat B(\tau_{\ell-1}))$ and
$(\tau_\ell,\widehat B(\tau_\ell))$. Therefore, for every integer
$h\in\{\tau_{\ell-1},\ldots,\tau_\ell\}$,
\[
B_h^{\mathcal I}
\leq
\frac{\tau_\ell-h}{\tau_\ell-\tau_{\ell-1}}
\widehat B(\tau_{\ell-1})
+
\frac{h-\tau_{\ell-1}}{\tau_\ell-\tau_{\ell-1}}
\widehat B(\tau_\ell)
\leq
\widehat B(h).
\]

Abel summation and the monotonicity of $\boldsymbol{v}$ imply
\begin{align*}
\boldsymbol{\beta}^{\mathcal I}\boldsymbol{v}^{T}
=\sum_{k=1}^{N-1}B_k^{\mathcal I}(v_k-v_{k+1})+B_N^{\mathcal I}v_N
\leq\sum_{k=1}^{N-1}\widehat B(k)(v_k-v_{k+1})+\widehat B(N)v_N
=\widehat{\boldsymbol{\beta}}\boldsymbol{v}^{T}.
\end{align*}

For the PAVA partition, the block boundaries are $0$, $N$, and the successive
kinks of $\widehat B$. These boundaries are contact points, and $\widehat B$
is affine between any two successive boundaries. Therefore, the average of
the original coefficients $\beta_j$ within each PAVA block equals the common
slope of $\widehat B$ on that block. Hence
$\boldsymbol{\beta}^{\mathrm{PAVA}}
=\widehat{\boldsymbol{\beta}}$. By the blockwise exchange identity,
\[
\boldsymbol{\beta}
\bigl(\boldsymbol{w}^{\mathrm{PAVA}}\bigr)^{T}
=
\boldsymbol{\beta}^{\mathrm{PAVA}}\boldsymbol{v}^{T}
=
\widehat{\boldsymbol{\beta}}\boldsymbol{v}^{T}.
\]

Lemma~\ref{lem:orderedpolytope} shows that every feasible expected prize vector
is a convex combination of adjacent coarsenings. Since every
such coarsening satisfies the upper bound derived above and the objective is
linear, the same bound applies to every
$\boldsymbol{w}\in\Wcal^{\downarrow}(\boldsymbol{v})$. Thus,
$\boldsymbol{w}^{\mathrm{PAVA}}$ is optimal, and the stated value formula
follows from the blockwise exchange identity.
\end{proof}

\section{Proofs for performance-contingent grading}\label{app:performanceproofs}

\subsection{Proof of Lemma~\ref{lem:continuousfeas}}

\begin{proof}
Monotonicity of $Q$ follows from incentive compatibility by
Lemma~\ref{lem:IC}. Fix $x\in[0,1]$ and let $S$ be the set of contestants
whose quantiles exceed $x$, with $K=|S|$. In every realization, the contestants
in $S$ receive $K$ prizes, whose total value cannot exceed the sum of the $K$
highest prizes. Hence, by symmetry,
\[
N\int_x^1 Q(s)\dd s
=
\E\left[\sum_{i\in S}\operatorname{prize}_i\right]
\leq
\E\left[\sum_{j=1}^{K}v_j\right].
\]
Under the fully assortative assignment, the contestants in $S$ are precisely
the $K$ highest-ranked contestants and therefore receive exactly the prizes
$v_1,\ldots,v_K$. Since their interim expected prize is $\Pbar$, symmetry also
gives
\[
\E\left[\sum_{j=1}^{K}v_j\right]
=
N\int_x^1\Pbar(s)\dd s.
\]
Combining the two relations proves the tail inequalities. At $x=0$, all
contestants are included and every prize is assigned, so
\[
\int_0^1 Q(s)\dd s
=
\int_0^1\Pbar(s)\dd s
=
\frac{1}{N}\sum_{j=1}^{N}v_j.
\]
Thus, $Q\preceq\Pbar$.
\end{proof}

\subsection{Proof of Lemma~\ref{lem:Phi}}

\begin{proof}
The function $a-(1-x)q(x)$ equals zero at $x=0$ and has derivative
\[
q(x)-(1-x)q'(x)=\varphi(x).
\]
It therefore equals $\int_0^x\varphi(s)\dd s$. At $x=1$ it equals $a$.
\end{proof}

\subsection{Proof of Theorem~\ref{thm:performance}}

\begin{proof}
\emph{Step 1: the upper bound.}
Let $Q$ satisfy the constraints in \eqref{eq:performanceprogram}. The tail
constraints and monotonicity imply the coordinate bounds
$v_N\leq Q(x)\leq v_1$ for every interior $x$. Indeed, subtracting a tail
inequality from the total-integral equality gives
\[
\int_0^xQ(s)\dd s
\geq
\int_0^x\Pbar(s)\dd s
\geq xv_N,
\]
and monotonicity gives
$xQ(x)\geq\int_0^xQ(s)\dd s$. The upper bound follows from
\[
(1-x)Q(x)
\leq
\int_x^1Q(s)\dd s
\leq
\int_x^1\Pbar(s)\dd s
\leq
(1-x)v_1.
\]
Thus, $Q$ is bounded and of bounded variation.

The function $\Phi-G$ is continuous, nonnegative, vanishes at both endpoints,
and has almost-everywhere derivative $\varphi-\phibar$. Integration by parts
gives
\begin{equation}\label{eq:firstboundproof}
\int_0^1(\varphi(x)-\phibar(x))Q(x)\dd x
=-\int_{[0,1]}(\Phi(x)-G(x))\dd Q(x)
\leq0.
\end{equation}

Define
\[
T(u):=\int_u^1\bigl(Q(x)-\Pbar(x)\bigr)\dd x.
\]
Since $Q\preceq\Pbar$, we have $T(u)\leq0$ for every $u\in[0,1]$.
Equality of the total integrals implies $T(0)=T(1)=0$, and
$T'(u)=-(Q(u)-\Pbar(u))$ almost everywhere. Integration by parts therefore
gives
\begin{equation}\label{eq:secondboundproof}
\int_0^1\phibar(x)\bigl(Q(x)-\Pbar(x)\bigr)\dd x
=
\int_{(0,1]}T(u)\dd\phibar(u)
\leq0.
\end{equation}
The inequality follows because $\dd\phibar$ is a nonnegative measure while
$T\leq0$. Combining \eqref{eq:firstboundproof} and
\eqref{eq:secondboundproof} yields
\begin{equation*}
	\begin{aligned}
		\int_0^1 \varphi(x)Q(x)\dd x
		\leq
		\int_0^1 \phibar(x)Q(x)\dd x
		\leq
		\int_0^1 \phibar(x)\Pbar(x)\dd x,
	\end{aligned}
\end{equation*}
which is \eqref{eq:twostepbound}.

Every no-exclusion contest outcome also satisfies $U(a)\geq v_N$.
Consequently, \eqref{eq:directTE} and \eqref{eq:twostepbound} imply
\[
\mathrm{TE}
\leq
N\int_0^1\phibar(x)\Pbar(x)\dd x-Na v_N.
\]

\emph{Step 2: $Q^*$ achieves the upper bound.}
For each ironing component $(\ell_m,r_m)$, monotonicity of $\Pbar$ implies
\[
\Pbar(\ell_m)
\leq
\frac{1}{r_m-\ell_m}
\int_{\ell_m}^{r_m}\Pbar(x)\dd x
\leq
\Pbar(r_m).
\]
The averages associated with successively higher grade cells are therefore
ordered, and they fit monotonically between the point-grade values of
$\Pbar$. Hence $Q^*$ is nondecreasing and
$Q^*(0)\geq\Pbar(0)=v_N$.

The first inequality in \eqref{eq:twostepbound} binds when $Q=Q^*$. Indeed,
$Q^*$ is constant on the interior of every ironing component, so the Stieltjes
measure $\dd Q^*$ assigns no mass there. Outside the ironing set, and at the
endpoints of every ironing component, $\Phi=G$. Therefore,
\[
\int_{[0,1]}(\Phi(x)-G(x))\dd Q^*(x)=0.
\]

The second inequality in \eqref{eq:twostepbound} also binds when $Q=Q^*$. For each ironing component
$(\ell_m,r_m)$, define
\[
\mu_m
:=
\frac{1}{r_m-\ell_m}
\int_{\ell_m}^{r_m}\Pbar(x)\dd x.
\]
By construction, $Q^*(x)=\mu_m$ almost everywhere on that component, and hence
\[
\int_{\ell_m}^{r_m}
\bigl(Q^*(x)-\Pbar(x)\bigr)\dd x
=
(r_m-\ell_m)\mu_m
-
\int_{\ell_m}^{r_m}\Pbar(x)\dd x
=
0.
\]
Moreover, $Q^*=\Pbar$ outside the operational grade cells. Since $\phibar$ is constant
on each ironing component, writing $\phibar_m$ for its value on $I_m$ gives
\[
\begin{aligned}
	\int_0^1
	\phibar(x)\bigl(Q^*(x)-\Pbar(x)\bigr)\dd x
	=
	\sum_{m\in\mathcal M}
	\phibar_m
	\int_{\ell_m}^{r_m}
	\bigl(Q^*(x)-\Pbar(x)\bigr)\dd x
	=0.
\end{aligned}
\]
Thus, both inequalities in \eqref{eq:twostepbound} bind at $Q^*$. Setting $U(a)=v_N$ gives
\[
\overline V^{\mathrm P}
=N\int_0^1\phibar(x)\Pbar(x)\dd x-Na v_N.
\]

\emph{Step 3: implementation of $Q^*$.}
Consider the graded assortative assignment in the theorem. If a contestant's
quantile $x$ lies outside $\bigcup_m I_m$, then $x$ forms a singleton point grade.
Since rival quantiles are independently drawn from a continuous distribution,
almost surely no rival makes the same point report. The number of rivals in
higher grades therefore equals the number of rivals with quantiles above $x$,
so the contestant has the same rank distribution as under the fully
assortative assignment. Her interim expected prize is consequently $\Pbar(x)$.

Now fix an operational ironing grade $I_m$ and let $\mu$ denote the common
interim expected prize of a contestant whose quantile lies in that grade. In
any realized profile, if $k$ contestants belong to the grade and $h$
contestants belong to higher grades, the grade collectively occupies ranks
$h+1,\ldots,h+k$, exactly as under full sorting. Uniform tie-breaking changes
only the internal assignment of the corresponding prize block. Taking
expectations gives
\[
N(r_m-\ell_m)\mu
=N\int_{\ell_m}^{r_m}\Pbar(s)\dd s,
\]
so
\[
\mu
=\frac{1}{r_m-\ell_m}
\int_{\ell_m}^{r_m}\Pbar(s)\dd s.
\]
Thus, the graded assortative rule induces $Q^*$.

Finally, as shown in Step 2, $Q^*$ is nondecreasing and
$Q^*(0)\geq v_N$. With $U(a)=v_N$,
Lemma~\ref{lem:IC} gives incentive compatibility, and
\[
E^*(a)=a\bigl(Q^*(0)-v_N\bigr)\geq0,
\qquad
\dd E^*(t)=t\dd\Qtilde^*(t)\geq0.
\]
Hence the effort schedule is nonnegative. On every ironing grade,
$\Qtilde^*$ and $E^*$ are constant. Across distinct grades,
$\Qtilde^*$ is strictly increasing because $\Pbar$ is strictly increasing;
therefore, $E^*$ is also strictly increasing. The induced expected total effort is exactly
\eqref{eq:Vperf}.
\end{proof}

\subsection{Proofs of Corollaries~\ref{cor:regular}
and \ref{cor:performancecoeff}}

\begin{proof}
If $J$ is nondecreasing, then $\varphi$ is nondecreasing, $\Phi$ is convex,
and $G=\Phi$. There is no ironing, so $Q^*=\Pbar$. Since
$\Pbar(0)=v_N$, the lowest equilibrium effort is
$
E^*(a)=a\bigl(\Pbar(0)-v_N\bigr)=0.
$

For the coefficient result, define
$
\gamma_j=N\int_0^1\phibar(x)p_j(x)\dd x.
$
The normalized density $Np_j$ is the Beta$(N-j+1,j)$ density. The likelihood ratio
$p_j/p_{j+1}$ is proportional to $x/(1-x)$ and is increasing, so the rank
densities are ordered by the monotone likelihood-ratio order, and thus by first-order stochastic dominance, from $j=1$ to
$j=N$. Since $\phibar$ is nondecreasing,
$\gamma_1\geq\cdots\geq\gamma_N$. By definition,
$
\overline\beta_j
=\gamma_j-Na\one\{j=N\},
$
so $\overline\beta_1\geq\cdots\geq\overline\beta_N$. Finally,
\[
\sum_{j=1}^{N}\overline\beta_j
=N\int_0^1\phibar(x)\dd x-Na
=N\bigl(G(1)-G(0)\bigr)-Na
=0.
\]
Substituting \eqref{eq:Pbar} into \eqref{eq:Vperf} yields
$\Vperf=\sum_j\overline\beta_jv_j$.
\end{proof}

\subsection{Proof of Theorem~\ref{thm:implementation}}

\begin{proof}
Suppose each rival of type $s$ chooses $E^*(s)$. By construction, the graded
ranking of equilibrium efforts coincides with the ranking obtained by pooling
abilities whose quantiles lie in the same ironing grade and ordering the
resulting grades by location. The induced interim prize is $Q^*$. The lowest
type's utility is
\[
Q^*(0)-\frac{a(Q^*(0)-v_N)}{a}=v_N.
\]

Consider a deviation by type $t$ to a passing effort $e'\geq e_{\min}$. If
$e'$ lies in a passing grade bracket, it produces the same graded rank as the
bracket's lower endpoint at a weakly higher cost. The lower endpoint equals
$E^*(s)$ for every type $s$ in that grade, so the deviation payoff is at most
the direct-mechanism payoff from reporting $s$, which is no greater than
the truthful utility $U^*(t)=\Qtilde^*(t)-E^*(t)/t$.

If $e'$ lies outside all brackets and is attained by a separating type $s$,
the deviation is exactly equivalent to reporting $s$ in the direct mechanism.
Consider next a possible unused gap immediately below a grade bracket whose
lower quantile endpoint satisfies $\ell_m>0$. Let $\underline e_m$ be the
lower endpoint of the bracket and define
\[
e_m^-
=
\lim_{x\uparrow\ell_m}E^*(q(x)),
\qquad
Q_m^-
=
\lim_{x\uparrow\ell_m}Q^*(x).
\]
If $e_m^-<\underline e_m$, any effort
$e'\in(e_m^-,\underline e_m)$ induces the limiting prize $Q_m^-$. For any
sequence $x_n\uparrow\ell_m$, direct incentive compatibility gives
\[
U^*(t)
\geq
Q^*(x_n)-\frac{E^*(q(x_n))}{t}.
\]
Taking limits yields
\[
U^*(t)
\geq
Q_m^- -\frac{e_m^-}{t}
>
Q_m^- -\frac{e'}{t}.
\]
If $e_m^-$ is itself unattained, the same limiting inequality weakly rules out
$e'=e_m^-$; if it is attained, that effort is already covered by the
separating-report case. When $\ell_m=0$, the interval below the bottom passing
bracket is the failing grade and is handled below.

Every unused effort above a pooled grade and below the next attained effort is
included in that grade's bracket. It therefore gives the same prize as the
lower endpoint at a strictly higher cost. If the highest grade is pooled, its
bracket is unbounded above, so every higher effort is similarly dominated. If
the top of the type distribution is separating, choosing more than $E^*(b)$
gives no better rank than $E^*(b)$ and costs strictly more.

Finally, consider $0\leq e'<e_{\min}$. When all rivals follow the equilibrium
strategy, the deviator is almost surely the unique contestant in the failing
grade and therefore receives the bottom prize $v_N$. Her deviation payoff is
\[
v_N-\frac{e'}{t}\leq v_N.
\]
Under the direct mechanism, $U^*(a)=v_N$ and
$U^{*\prime}(t)=E^*(t)/t^2\geq0$ almost everywhere, so
$U^*(t)\geq v_N$ for every type. Thus, failure is not profitable, and zero
effort is the cheapest action within the failing grade. These cases exhaust
all unilateral deviations.
\end{proof}

\IfFileExists{econ.bst}
  {\bibliographystyle{econ}}
  {\bibliographystyle{plainnat}}
\bibliography{bib_grading}

\end{document}